\documentclass{article}

\usepackage{arxiv}

\usepackage[utf8]{inputenc} 
\usepackage[T1]{fontenc}    
\usepackage{hyperref}       
\hypersetup{colorlinks=true, linkcolor=gnblue4, breaklinks=true, urlcolor=gnblue4, citecolor=gnblue4}

\usepackage{url}            
\usepackage{booktabs}       
\usepackage{amsfonts}       
\usepackage{nicefrac}       
\usepackage{microtype}      
\usepackage{lipsum}		
\usepackage{graphicx,xcolor, subcaption}
\usepackage{sidecap, caption}
\usepackage{natbib}
\usepackage{doi}

\usepackage{bbm}
\usepackage{mathbbol}
\usepackage{amsmath}
\usepackage{amsthm}
\usepackage{imakeidx}

\newtheorem{theorem}{Theorem}

\newtheorem{definition}{Definition}

\newcommand{\real}{\mathbb{R}}
\newcommand{\subscr}[2]{#1_{\textup{#2}}}
\newcommand{\En}[1]{\mathrm{E}}

\DeclareSymbolFont{bbold}{U}{bbold}{m}{n}
\DeclareSymbolFontAlphabet{\mathbbold}{bbold}

\newcommand{\vect}[1]{\mathbbold{#1}}
\newcommand{\vectorones}[1][]{\vect{1}_{#1}}
\newcommand{\vectorzeros}[1][]{\vect{0}_{#1}}

\definecolor{gnblue4}{RGB}{0,108,212} 

\title{Input-Driven Dynamics for Robust Memory Retrieval in Hopfield Networks}

\author{ \href{https://orcid.org/0009-0000-3444-0838}{\includegraphics[scale=0.06]{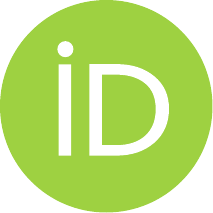}\hspace{1mm}Simone Betteti} \\
	Department of Information Engineering\\
	Università degli Studi di Padova\\
	Padova, 35131, IT \\
	\texttt{bettetisim[at]dei.unipd.it} \\
	\And
	\href{https://orcid.org/0000-0002-9439-296X}{\includegraphics[scale=0.06]{orcid.pdf}\hspace{1mm}Giacomo Baggio} \\
	Department of Information Engineering\\
	Università degli Studi di Padova\\
	Padova, 35131, IT \\
	\texttt{baggio[at]dei.unipd.it} \\
	\AND
    \href{https://orcid.org/0000-0002-4785-2118}{\includegraphics[scale=0.06]{orcid.pdf}\hspace{1mm}Francesco Bullo} \\
	Center for Control, Dynamical Systems and Computation\\
	University of California at Santa Barbara,\\
	Santa Barbara, CA, 93106 IT \\
	\texttt{bullo[at]ucsb.edu}
    \And
    \href{https://orcid.org/0000-0001-8926-1888}{\includegraphics[scale=0.06]{orcid.pdf}\hspace{1mm}Sandro Zampieri} \\
	Department of Information Engineering\\
	Università degli Studi di Padova\\
	Padova, 35131, IT \\
	\texttt{zampi[at]dei.unipd.it}\\ 
}

\renewcommand{\headeright}{}
\renewcommand{\undertitle}{}
\renewcommand{\shorttitle}{Input-Driven Dynamics for Robust Memory Retrieval in Hopfield Networks}

\hypersetup{
pdftitle={A template for the arxiv style},
pdfsubject={q-bio.NC, q-bio.QM},
pdfauthor={Simone Betteti, Giacomo Baggio, Francesco Bullo, Sandro Zampieri},
pdfkeywords={Hopfield Networks, Associative Memory, Short-Term Plasticity, Dynamical Systems},
}

\begin{document}
\maketitle

\begin{abstract}
The Hopfield model provides a mathematically idealized yet insightful framework for understanding the mechanisms of memory storage and retrieval in the human brain.  This model has inspired four decades of extensive research on learning and retrieval dynamics, capacity estimates, and sequential transitions among memories.  Notably, the role and impact of external inputs has been largely underexplored, from their effects on neural dynamics to how they facilitate effective memory retrieval. To bridge this gap, we propose a novel dynamical system framework in which the external input directly influences the neural synapses and shapes the energy landscape of the Hopfield model. This plasticity-based mechanism provides a clear energetic interpretation of the memory retrieval process and proves effective at correctly classifying highly mixed inputs. Furthermore, we integrate this model within the framework of modern Hopfield architectures, using this connection to elucidate how current and past information are combined during the retrieval process.  Finally, we embed both the classic and the new model in an environment disrupted by noise and compare their robustness during memory retrieval.
\end{abstract}


\section{Introduction}
Since the beginning of the 80s, the words ``Associative Memory Network" have  closely echoed with ``Hopfield Network" \citep{hopfield1982neural,hopfield1984neurons}, and a plethora of subsequent works have endeavored to provide a detailed picture of the properties of such networks \citep{amit1987Stat,crisanti1986saturation,treves1988metastable}. Drawing from the toolbox of statistical mechanics, Hopfield networks provided a convincing explanation for the multi-stability of memories as function of the neurons couplings, and therefore a plausible, dynamic retrieval mechanism over an energy landscape. Recently, in a machine-learning driven Renaissance for associative memory networks, the original framework has been generalized to higher order interactions \citep{krotov2016dense} and to multi-layered architectures \citep{krotov2020large, chaudhry2023long}, thus endowing the model with both a significantly improved capacity \citep{demircigil2017model} and a direct bridge to state-of-the-art transformer models and their attention mechanism \citep{ramsauer2021hopfield}. Moreover, the new framework has paved the way for new hypotheses on how neurons and astrocytes could interact \citep{kozachkov2023building}, at the functional level, to support cognitive processes. The effort to bridge formal approaches and neuroscience is of paramount importance for the advancement of both fields. As proposed in \citep{treves1992computational}, attractor dynamics may be a key component of hippocampal functioning, where the signal relayed by cortical areas is sparsified and orthogonalized in the CA3-CA1 regions \citep{yassa2011pattern, rolls2013mechanisms}. In addition, simple attractor models provide a viable tool to study global cortical dynamics in the brain \citep{russo2012cortical, naim2018reducing}, by partitioning the surface in interacting patches of cortex each idealized by Hopfield like networks.

In classic treatments on computational neuroscience \citep{amit1989modeling, dayan2005theoretical, gerstner2014neuronal}, memory retrieval in the Hopfield model is implicitly described as a two-step process. 
First, a noisy or incomplete input is presented as a cue and adopted as an initial condition. 
Then, driven by an energy landscape, the network state flows towards the closest energy minimum
representing the prototypical memory. 
The literature however lacks an explanation for how an external input becomes an initial condition in the neural dynamics; it is worth noting that
external inputs and initial conditions play distinct roles in the behavior of a dynamical system.
Most importantly, while this classic two-step process is natural within an algorithmic paradigm, it fails to explain how neuronal circuits continuously react and adapt in real time to external inputs.

In light of these limitations, we advocate for a paradigm shift from 
a two-step mechanism, akin to a standard algorithmic approach,
to a input-driven dynamic mechanism, aligned with the principles of online algorithms and continual learning \citep{zenke2017continual, hadsell2020embracing, lesort2020continual}. 
To this extent, we propose a novel version of the Hopfield model that is driven by external inputs. 
A key feature of this model is that the input shapes the energy landscape and affects the resulting gradient descent flow (see Fig.~\ref{fig:EnerClass}).
Furthermore, our model admits a simple representation as
a modern Hopfield network \citep{krotov2020large, ramsauer2021hopfield};
this representation provides a conceptual bridge with the recent
literature on transformer models and machine learning.  Finally, the
addition of noise reveals the advantageous integration of past and
present information by our model, thereby reducing misclassification
errors induced by inconsistent or `glitchy' inputs.

\begin{figure}[t!]
\centering
\includegraphics[width=11.4cm]{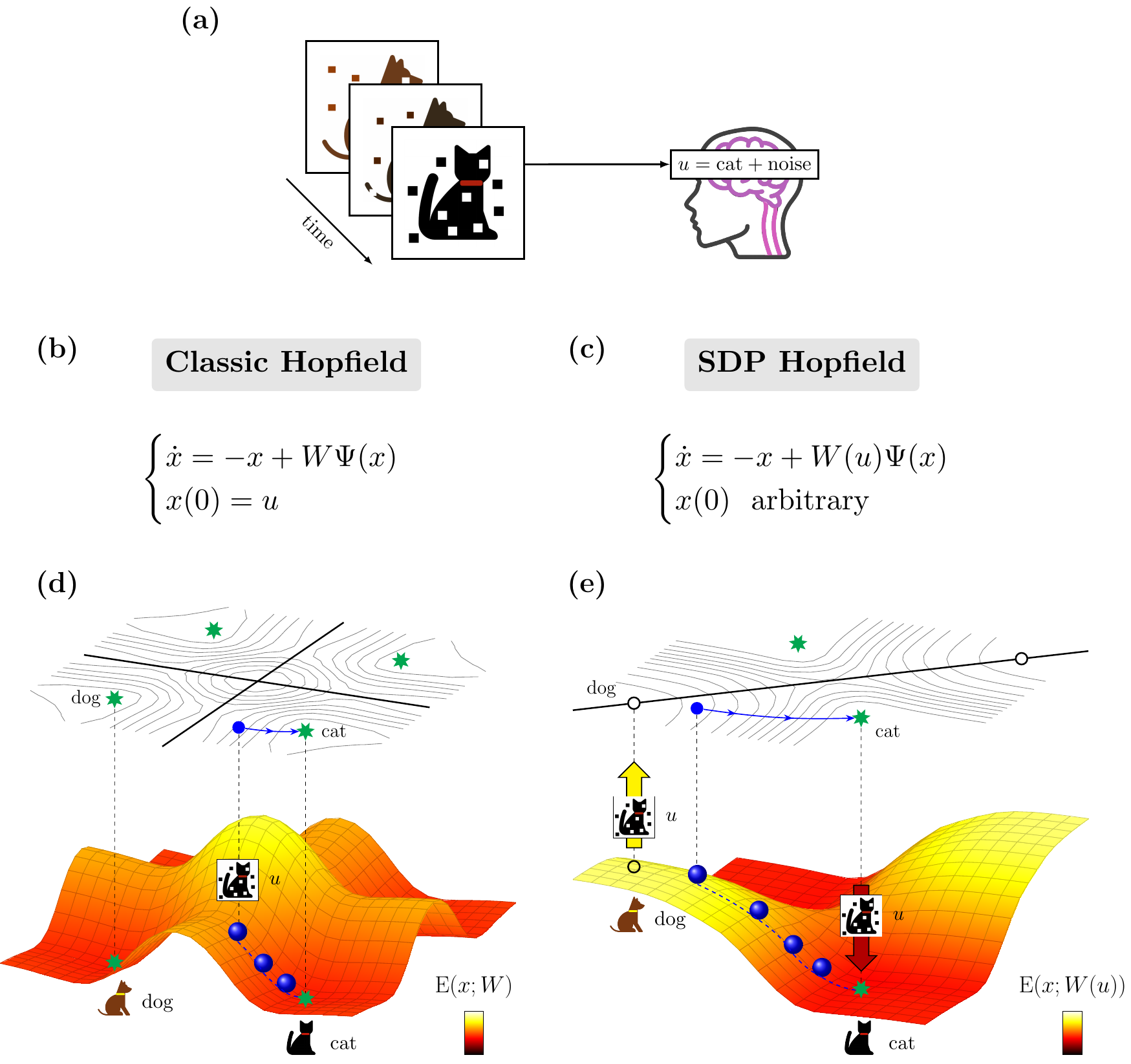}
\caption{Comparison between classic Hopfield and IDP Hopfield models. (a) A slowly morphing sequence of noisy images is presented as a input to the observer, who updates its belief state to retrieve the memory closest to the current image $u$. This adaptation process occurs continuously. 
(b) In the classic model, the network state is set to an initial condition $x(0)$ equal to the current image $u$ and then the Hopfield dynamics performs the memory retrieval task. 
(c) In the proposed input-driven plasticity model, the network initial condition is arbitrary, the image $u$ modifies the synaptic weights $W(u)$, and the Hopfield dynamics with modified synaptic weights performs the memory retrieval task. This dynamics is well posed and naturally tracks the morphing images also when the image is time-varying $u=u(t)$.
(d) In the classic model, the Hopfield dynamics is a gradient descent for the energy $\mathrm{E}(x;W)$:
the blue ball, representing the neural state, rolls from an initial condition towards a stable minimum point (cat memory).
Therefore, the retrieval process is successful when the noisy image $x(0)=u$ (dotted cat) lies in the region of attraction of the correct memory (cat memory). 
(e) In the proposed model, the noisy image $u$ directly modifies the synaptic matrix $W(u)$ and in turn the energy landscape $\mathrm{E}(x;W(u))$, thereby extending the region of attraction of the correct memory. The retrieval process is successful from generic initial conditions when the correct memory is the unique minimum of the landscape. Specifically, in the panel, the noisy image (dotted cat) renders the correct memory a minimum (cat memory) and the incorrect memory (dog memory) no longer an equilibrium of $\mathrm{E}(x;W(u))$. \vspace{1cm}}
\label{fig:EnerClass}
\end{figure}

\section{Results}
\subsection{Primer on Hopfield Networks}

Hopfield networks are a fundamental tool in the study of high level, distributed memories retrieval \citep{hopfield1982neural}. They significantly simplify neural dynamics into the interplay of two components: a dissipative flow, constantly polarizing the network state towards its resting value, and a synaptic flow, which takes into account the weighted sum of the incoming activity from other neurons in the network. Namely,
\begin{equation}
    \begin{cases}
        \dot{x}(t) = \underbrace{-x(t)}_{\text{dissipation}}\enspace+\enspace \underbrace{ W\Psi(x(t))}_{\text{synaptic interconnections}}\\
        x(0) = x_{0}\in\mathbb{R}^{N}
    \end{cases}
\end{equation}
where the prototypical memories $\{\xi^\mu\}_{\mu=1}^P$, $\xi^\mu\in\{-1,+1\}^N$ are assumed to be orthogonal or with entries that are independent and identically distributed. They are stored in the synaptic matrix $W$ through one-shot Hebbian learning as 
\begin{equation}
    W = \frac{1}{N}\sum_{\mu=1}^{P}\xi^{\mu}{\xi^{\mu}}^{\top}
\end{equation}
Under suitable assumptions on the activation function $\Psi$ \citep{krotov2020large}, convergence to any of the stored memories is guaranteed \citep{hopfield1984neurons} by the existence of the energy function (see Fig.~\ref{fig:EnerClass}(d))
\begin{equation}
    \mathrm{E}\big{(}x;W\big{)} = -\frac{1}{2}\Psi(x)^{\top}W\Psi(x)+x^{\top}\Psi(x)-\sum_{i=1}^{N}\int_{0}^{x_{i}}\Psi_{i}(z)\:dz
\end{equation}
given an appropriate initial condition $x_{0}$.

The model, first proposed as a discrete time dynamic system \citep{hopfield1982neural}, has been particularly successful and captivating at explaining pattern reconstruction starting from a partial or corrupted cue \citep{gerstner2014neuronal,dayan2005theoretical}. In its simplicity, the original theory effectively framed memory retrieval in the cascade of reactions that lead a network of elementary computational units to fix in a meaningful collective state, i.e.~the prototypical memory. Furthermore, the geometric constraints placed on the prototypical memories and on the synaptic matrix allowed for a detailed study of the network capacity. In the original paper and subsequent works \citep{hopfield1982neural,Amit1987Inf} it was estimated at around $0.14N$, and only later was refined \citep{petritis1995thermodynamic} at $\frac{N}{6{\log(N)}}$ for exact convergence in probability. These estimates represent the number of patterns that can be recovered in the asymptotic limit of infinite network size without the endogenous noise disrupting their stability. 

The original model has then been expanded to the continuous setting \citep{hopfield1984neurons} with the inclusion of external inputs from the environment, according to the following equation
\begin{equation}\label{CHC}
    \dot{x}(t) = -x(t) + W\Psi(x(t))+u(t) 
\end{equation}
The introduction of an external input conflicts with the previous definition of prototypical memories as stable patterns of activity encoded in $W$. If the input is time-varying, then the energy function is no longer Lyapunov, namely it is not decreasing along the system trajectories,
and hence we may lose convergence to any stable pattern. On the other hand, if the input is constant, it may distort the energy landscape in such a way that the minima are not related to any of the prototypical memories. Indeed, when the input is a heterogeneous mixture of memories, the energy landscape is inevitably corrupted and none of the stored memories is retrievable (see Fig.~\ref{fig:det_comp}(b)). A possible solution is to clamp the input \citep{Battaglia1998} for a brief time interval $\delta t$ as an external driver for the dynamics, then un-clamp it and let the dynamics naturally evolve, namely
\begin{align}
    \dot{x}(t)&=-x(t)+W\Psi(x(t))+u(t)\zeta(t)\label{CHM}\\
    \zeta(t)&\propto\mathbbm{1}_{[0,\delta{t}]}(t)\label{Mod}
\end{align}
where $\mathbbm{1}_{S}(t)$ is the characteristic function associated to any subset $S$ of the real numbers defined by letting $\mathbbm{1}_{S}(t)=1$ if $t\in S$ and $\mathbbm{1}_{S}(t)=0$ otherwise. The external driver shoots the state trajectory in the direction of the new input, acting as bias for the next retrieval. The model~\eqref{CHM} is successful at discriminating the dominant component of very mixed inputs (see Fig.~\ref{fig:det_comp}(c)), but it does so at the cost of introducing network-wide synchrony. In addition, as observable from Fig.~\ref{fig:det_comp}(b,c), each new input instantaneously alters the network activity, canceling the information about any previous activity. Instead, in Fig.~\ref{fig:det_comp}(d) our novel framework displays a remarkable capability of successfully retrieving the correct memory given the continuous external input, and it will be presented in the next section.   

\begin{SCfigure*}[\sidecaptionrelwidth][tbph!]
\centering
\includegraphics[width=.5\linewidth]{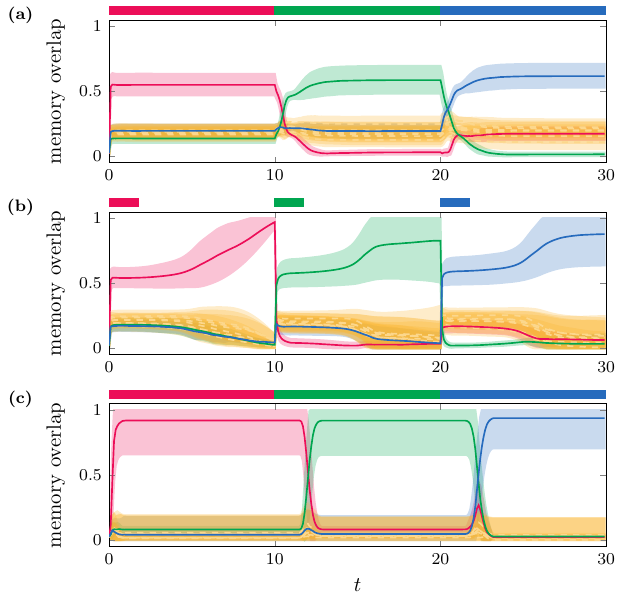}
\caption{Exploration of the response of different associative memory models to a time-varying input. Each retrieval plot displays a sample average of the memory overlap parameter over $50$ retrieval tasks, with the associated standard deviation shown through the shaded, color-matched surrounding area. Panel~(a),~(b), and (c) show the response of three different models to a time-varying input. The simulation time is segmented, with each sub-interval featuring a constant input. The horizontal bars on top of each retrieval plot display the duration of the external input and the color associated to its dominant saliency weight. At each external input switch, the previous dominant saliency weight shrinks below the stability threshold $\alpha_{\text{stability}}$. 
(a) Classic Hopfield model~\eqref{CHC}. The network dynamics converge to a mixed state, precluding exact retrieval of any individual memory. 
(b) Input-modulated Hopfield model~\eqref{CHM}. The modulator~\eqref{Mod} shuts off the input during the non-shaded sub-intervals. The network dynamics then freely recall the prototypical memory associated with the dominant component in the input.
(c) IDP Hopfield model~\eqref{MHC}. The network naturally recalls the prototypical memory associated with the dominant component in the input. Remarkably, the model exhibits memory retention even after input switching, showcasing an intriguing integrative feature for past and new information.}
\label{fig:det_comp}
\end{SCfigure*}

\subsection{The Input-Driven Plasticity (IDP) Hopfield Model}
The aim of this work is that of studying a plausible mechanism that, given a mixed input, first maintains fixation for short transients on what was previously retrieved, that is past information, and then gradually merges it with the current input, that is present information. This gradual integration should, at a certain point, favor the present information and retrieve the new correct memory. We thereby introduce a novel, externally modulated Hopfield model, named input-driven plasticity (IDP) Hopfield model, that tries to capture this exact phenomenology. 
The new IDP Hopfield model is defined as
\begin{equation}\label{MHC}
     \dot{x}(t)=-x(t)+W(u(t))\Psi(x(t))
\end{equation}
where for a generic input $u(t)$ the novel \emph{input modulated synaptic matrix} is
\begin{align}
     W(u(t)) &= \frac{1}{N}\sum_{\mu=1}^{P}\xi^{\mu}{\xi^{\mu}}^{\top}u(t){\xi^{\mu}}^{\top}\nonumber\\
     &=\frac{1}{N}\sum_{\mu=1}^{P}\alpha_{\mu}(t)\xi^{\mu}{\xi^{\mu}}^{\top}
\end{align}
where we call the $\alpha_{\mu}(t):={\xi^{\mu}}^{\top}u(t)$ saliency
weights in accordance with previous literature
\citep{blumenfeld2006dynamics,tang2010memory}.
These previous studies have explored the role of saliency weights
and their effect on the dynamics~\eqref{MHC} within the context of
the morphing problem, focusing on how these weights drive
transitions between memories as they evolve. However, a
comprehensive analysis of how the magnitude
of saliency weights affects the existence and stability of memories as equilibrium points remains missing. To address this gap, we begin by characterizing the input, followed by a detailed examination of the specific conditions on saliency weights necessary to ensure the existence and stability of memories as equilibria for \eqref{MHC}.
We say that an input is homogeneous when it is significantly aligned
with one of the prototypical memories and almost orthogonal to the others,
i.e.~when there exists $\alpha_{\nu}\gg{0}$ and $\alpha_{\mu}\approx{0}$
for all $\mu\neq{\nu}$. Conversely, we say that the input is
well-mixed when it has similar saliency weights associated to all the
prototypical memories, i.e.~$\alpha_{\mu}\approx{\alpha_{\nu}}$ for all
$\mu,\nu$. Furthermore, under a timescale separation hypothesis for neural
dynamics and biologically trackable inputs, i.e. under the hypothesis that
neural dynamics are much faster than input dynamics, the latter can be
considered constant throughout each retrieval time interval.
Unlike Hopfield’s foundational works
\citep{hopfield1982neural, hopfield1984neurons} based upon a static
synaptic matrix, we build on insights from \citep{blumenfeld2006dynamics,
mongillo2008synaptic} to introduce elementary dynamics into the
synaptic structure. This dynamic approach yields advantageous properties
for system behavior, particularly by demonstrating how externally
modulated synaptic changes can promote the retrieval of certain patterns
while inhibiting others.  To rigorously quantify how the input determines
specific retrieval properties at the network level, we analytically study
the existence and stability of memories, specifically for the case of a
diagonal monotonic homogeneous activation function
$\Psi(x)=(\psi(x_{1}),\dots,\psi(x_{N}))$:
\begin{enumerate}
    \item given a prototypical memory $\xi^{\mu}$, the associated retrievable memory is a vector of the form $x_{\mu}=\gamma_{\mu}\xi^{\mu}$ for $\gamma_{\mu}\neq{0}$ and is an equilibrium of \eqref{MHC}
    \item a retrievable memory exists under the current input if and only if the associated saliency weight is larger than an existence threshold:
    \begin{equation}
        \alpha_{\mu}> \subscr{\alpha}{existence}=1
    \end{equation}
    and in this case $\gamma_{\mu}$ satisfies
    \begin{equation}
        \frac{\gamma_{\mu}}{\alpha_{\mu}}=\psi(\gamma_{\mu})
    \end{equation}
    \item if multiple retrievable memories exist under the current input, then there exists a stability threshold $\subscr{\alpha}{stability}>1$ such that each memory $x_{\mu}$ is stable if and only if $\alpha_{\mu}>\subscr{\alpha}{stability}$. 
\end{enumerate}

\begin{figure}[tbph!]
    \centering
    \includegraphics[width=11.4cm]{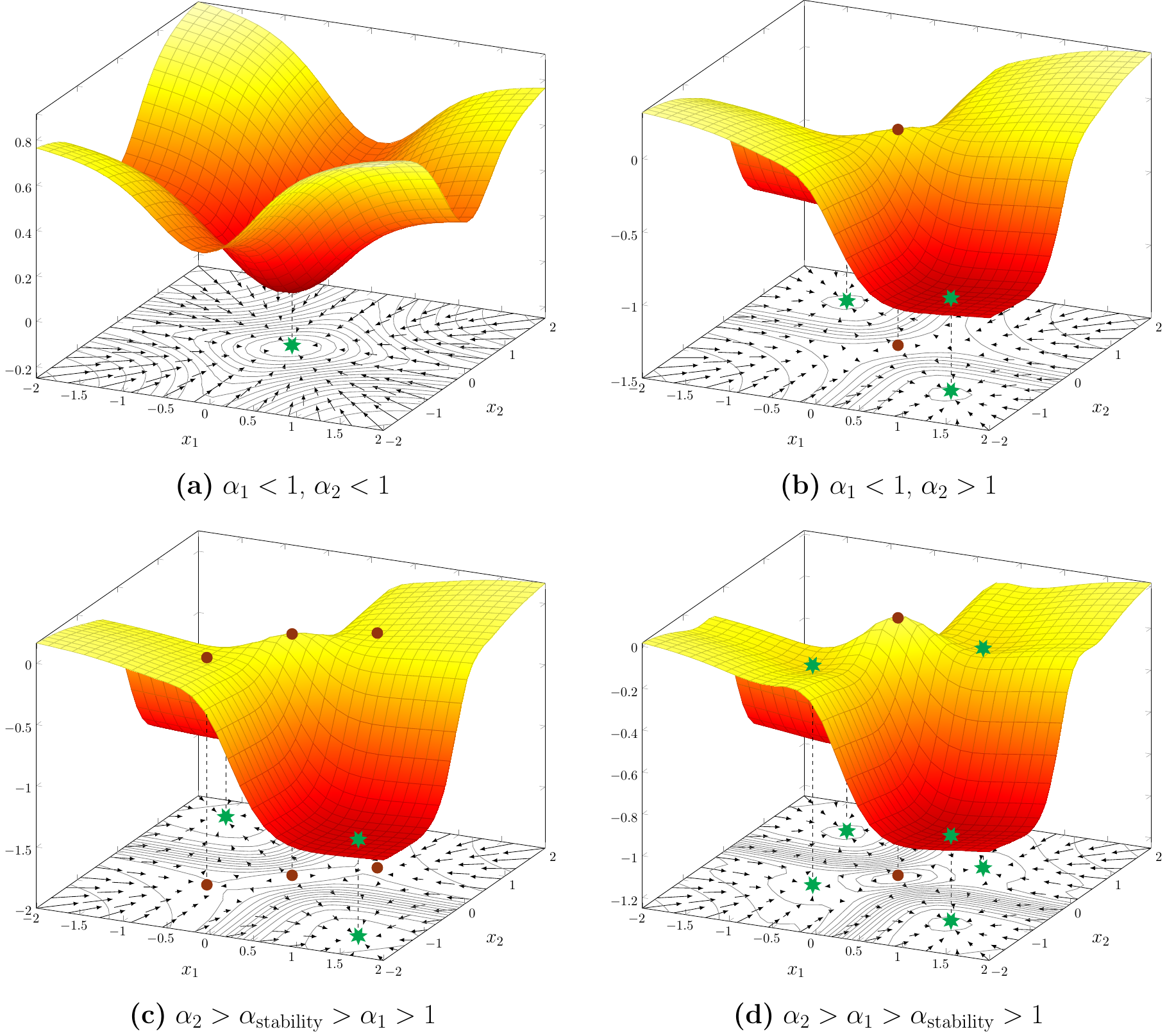}
    \caption{Energy landscapes for IDP Hopfield model for varying saliency weights. Stable and unstable equilibria are depicted as green stars and brown dots, respectively. Recall the existence threshold is $\subscr{\alpha}{existence}=1$ and, when 
    multiple memories exist in the input, the stability threshold satisfies $\subscr{\alpha}{stability}>1$. 
    (a) ``no memories'' $\alpha_1<1$, $\alpha_2<1$: when no memory is sufficiently strong in the input, the only global minimum is at the origin and it is globally attractive for the dynamics. This situation corresponds to a confusion state for the network, in which the input is not strong enough to evoke any retrieval. 
    (b) ``one memory'' $\alpha_1<1$, $\alpha_2>1$: when the saliency weight for precisely one memory is above $\subscr{\alpha}{existence}$, two symmetric equilibria appear corresponding to the memory and they are attractive from almost all initial conditions. In this case, the origin is a saddle point.
    (c) ``one stable and one unstable memory'' $\alpha_2>\subscr{\alpha}{stability}>\alpha_1>1$: two symmetric equilibria appear corresponding to the stable memory and they are attractive from almost all initial conditions. The other two symmetric equilibria are saddle points and the origin becomes an unstable maximum. 
    (d) ``two stable memories'' $\alpha_{2}>\alpha_{1}>\subscr{\alpha}{stability}$: four symmetric stable equilibria appear, corresponding to the two memories of the model. The memory associated to the dominant saliency weight carves deeper valley in the energy landscape than the other memory. When $\alpha_{1}\approx\subscr{\alpha}{stability}$, the shallowness of the valley associated to the first memory facilitates  outward jumps due to stochastic fluctuations.\vspace{1cm}}
    \label{fig:EnergyCases}
\end{figure}

Notice that the classic Hopfield model is a  specific sub-case of the IDP model, where the input impinging on the synaptic matrix is characterized by $\alpha_{\mu}\equiv{k}>1$, for all $\mu$. 
As clearly observable from Fig.~\ref{fig:EnergyCases} the external input induces a smooth reshaping of the energy landscape. In particular, the modified energy function is
\begin{equation}\label{eq: energy}
    \mathrm{E}\big{(}x;W(u)\big{)} = -\frac{1}{2}\Psi(x)^{\top}W(u)\Psi(x)+x^{\top}\Psi(x)-\sum_{i=1}^{N}\!\int_{0}^{x_{i}}\!\!\!\psi(z)\:dz
\end{equation}
and for every $\nu\in\{1,\dots,P\}$ such that $\alpha_{\nu}<\alpha_{\mu}$ we have that
\begin{equation}
    \mathrm{E}(x_{\mu},u)<\mathrm{E}(x_{\nu},u)
\end{equation}
Exploiting the timescale separation hypothesis for neural dynamics and input dynamics, we also have that 
\begin{align}
    \frac{d}{dt}\mathrm{E}\big{(}x(t);W(u(t))\big{)}&=\frac{\partial{\mathrm{E}}}{\partial{x}}\big{(}x(t);W(u(t))\big{)}\dot{x}(t)\nonumber\\
    &+\frac{\partial{\mathrm{E}}}{\partial{u}}\big{(}x(t);W(u(t))\big{)}\underbrace{\dot{u}(t)}_{\approx{0}}\nonumber\\
    &=\frac{\partial{\mathrm{E}}}{\partial{x}}\big{(}x(t);W(u)\big{)}\dot{x}(t)\leq{0}
\end{align}
The new, IDP Hopfield model naturally endows us with a direct interpretation of when the input is clearly understandable, and when instead is too vague to evoke any kind of retrieval. As previously presented, memories exist (are retrievable) only if their respective saliency weights in the input are at least unitary. Thus, for any input that is well-mixed and such that $\alpha_{\mu}<1$ for all ${\mu}$ (Fig.~\ref{fig:EnergyCases}(a)) the dynamics will converge to the origin, which we call confusion state. This state exists only in the new IDP model, as the classic model always has all the retrievable memories irrespective of the input. The possibility of falling into a confusion state on the basis of what is experienced seems quite plausible, as the clarity of what we sense in our everyday experience has a clear bearing on what we recall.

The reader is referred to the SI for a detailed proof of the presented results and for details on the numerical characterization of the inputs used for the simulations.

\subsection{A Modern Interpretation}
As we have mentioned in the introduction, the IDP model naturally lends itself to a description through the modern formalism outlined in \citep{krotov2020large, ramsauer2021hopfield}. By means of this formalism, a recurrent neural network such as the Hopfield model can be effectively deconstructed into two interacting layers void of intralayer connections.

The following tripartite architecture characterizes a more general model that captures the interaction of the activity in the memory layer $y\in{\mathbb{R}^{P}}$ and in the saliency layer $\alpha\in{\mathbb{R}^{P}}$. The new combined information is then exploited to drive the retrieval in the feature layer $x\in\mathbb{R}^{N}$. In summary, the modern Hopfield reformulation of the IDP model is 
\begin{align}
    \tau_{x}\dot{x}(t) &= -x(t) + M_{x}\Psi_{x}(y(t) \odot\alpha(t))\label{featl}\\
    \tau_{y}\dot{y}(t) &= -y(t) + M_{y}\Psi_{y}(x(t))\label{memol}\\
    \tau_{\alpha}\dot{\alpha}(t) &= -\alpha(t) + M_{\alpha}\Psi_{\alpha}(u(t))\label{alphal}
\end{align}
where the symbol $\odot$ is the Hadamard entrywise product, namely $(y(t)\odot\alpha(t))_i=y_{i}(t)\alpha_{i}(t)$, and $\Psi_{\alpha}$ is an activation function implementing either a linear or non-linear processing of the input.
Notice that when $M_{x} = M$ and $M_{y}=M_{\alpha}=M^{\top}$ with $M=N^{-\frac{1}{2}}[\xi^{1}\cdots\xi^{P}]\in{N}^{-1/2}\{-1,1\}^{N\times{P}}$, and $\Psi_{x},\Psi_{\alpha}$ are identity functions, the equations~\eqref{featl},\eqref{memol}, \eqref{alphal} reduce to the IDP Hopfield model~\eqref{MHC} in the limits $\tau_{y}\rightarrow{0}$ and $\tau_{\alpha}\rightarrow{0}$.
Unraveling the activity of the IDP Hopfield network into the distinct components allows for a better qualitative understanding of the layers' contribution (see Fig.~\ref{fig:block} for a block representation). The memory layer serves the function of pooling layer, projecting the activity of the feature layer into a similarity space. The pooled activity is then modulated by the input decomposition. It is clear then that the combination of pooling and modulation implements a natural trade-off between past internal activity and externally incoming information.     
\begin{figure}[tbph!]
    \centering
    \includegraphics[width=1\linewidth]{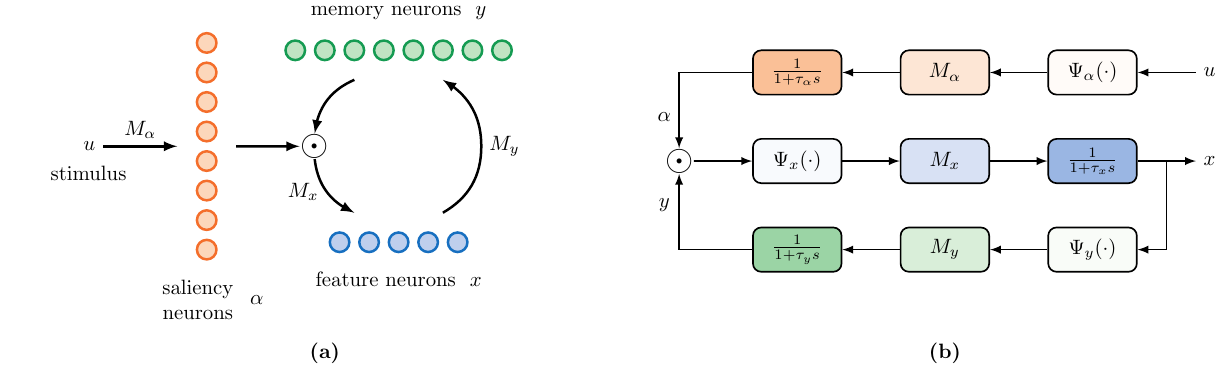}
    \caption{Illustration of the 
    modern Hopfield reformulation~\eqref{featl},\eqref{memol},\eqref{alphal} for the IDP Hopfield model with input filter. 
    The symbol~$\odot$ denotes a Hadamard entrywise product.
    (a) Neural network representation with interconnected layers and synaptic weights $M_\alpha$, $M_x$, and $M_y$. 
    (b) Block diagram representation, where each block~$\frac{1}{1+\tau s}$, where $\tau=\tau_\alpha,\tau_x,\tau_y$, denotes a low-pass filter with cutoff frequency~$1/\tau$, each block $M_\alpha$, $M_x$, and $M_y$ denotes a matrix-vector multiplication, and each block $\Psi_{\alpha}(\cdot)$, $\Psi_{x}(\cdot)$, and $\Psi_{y}(\cdot)$ denotes an activation function.}
    \label{fig:block}
\end{figure}
In the model discussed so far we have the input processing $M_{\alpha}\Psi_{\alpha}(u(t))=M^{\top}u(t)$, because it allows a direct logical and visual interpretation of the relationship between the input decomposition and the network response. However, in a realistic neuronal system the relationship may not be so simple and interpretable. Indeed, the most general case would be that where $M_{\alpha}\in{\mathbb{R}^{P\times{K}}}$ and $u(t)\in{\mathbb{R}^{K}}$ is a generic input. This is the case when the input originates from a different cortical region possibly having a different output dimensionality, and therefore the relation with the prototypical memories is more subtle. A discussion of how the $\Psi_{\alpha}$ output channels are related to the memories stored in the recurrent network is beyond the scope of this article, but remains an important problem open to original solutions.

\subsection{IDP Induces Resiliency to Noise} Both the classic~\eqref{CHM} and the IDP Hopfield~\eqref{MHC} models are fully capable of retrieving the correct memory when subjected to the same well-mixed input (see Fig.~\ref{fig:det_comp}(c,d)). This may mislead the reader to conclude that the two models are in fact equivalent. However, both are still ideal models, where neural units are perfectly insulated. Thus, they fail to capture the essence of real neural phenomena in settings with ubiquitous background noise. 
Thus, we consider the associated stochastic dynamics
\begin{align}\label{SDE}
    \dot{x}(t) &= \mathcal{F}(x(t),u)+ \eta(t)\\
    \mathbb{E}[\eta(t)]&=0\\
    \mathbb{E}[\eta(t)\eta(\tau)^{\top}]&=\sigma^{2}I_{N}\delta(t-\tau)
\end{align}
where the drift term $\mathcal{F}$ is in place of a generic field chosen between the classic and IDP models, and the shot-noise increments are taken to be state-independent for simplicity. 
It is well known \citep{yan2013nonequilibrium,brinkman2022metastable,moss1989noise} that the probability density function (P.D.F.) $\mathbb{P}(t,x)$ associated to~\eqref{SDE} evolves according to the parabolic Fokker-Planck P.D.E. (F.P.E.)
\begin{equation}\label{FPE}
        \frac{\partial{\mathbb{P}(t,x)}}{\partial{t}}=-\sum_{i=1}^{N}\frac{\partial}{\partial{x_{i}}}[\mathcal{F}_{i}(x)\mathbb{P}(t,x)]\!+\!\frac{\sigma^{2}}{2}\!\sum_{i=1}^{N}\frac{\partial^{2}}{\partial{x_{i}^{2}}}\mathbb{P}(t,x)
\end{equation}
Given the close dependence of the drift term $\mathcal{F}(x)=D\Psi^{-1}(x)\nabla_{x}\mathrm{E}(x,W)$ on the gradient of the energy~\eqref{eq: energy}, the mass of the stationary probability density is expected to concentrate around the deepest minima in the energy landscape. Numerical results confirm this intuition. Specifically, in Fig.~\ref{fig:PDFs} we employ a finite difference scheme in a bounded domain to compute the stationary solution of~\eqref{FPE} for both the classic and IDP Hopfield models. As expected, the resulting densities exhibit peaks close to the minima of the corresponding energy functions. In addition, the height of each peak is proportional to the depth of the corresponding minimum in the energy landscape. In simple terms, these results indicate that the wandering of the stochastic dynamics will most likely be confined to the valley associated with the deepest energy minima.
 
To further strengthen the link between energy and stationary densities, we observe that the stochastic model~\eqref{SDE} is similar to the known diffusion state machine \citep{wong1991stochastic}, which admits the clean energy-dependent stationary P.D.F.
\begin{equation}\label{StatDS}
    \mathbb{P}_{\infty}(x)=Z^{-1}e^{-\frac{\mathrm{E}(x; W)}{\sigma^{2}}}
\end{equation}
In the SI we numerically show that~\eqref{StatDS} and the stationary density of~\eqref{SDE} exhibit a similar shape.
\begin{figure}[tbph!]
\centering
\includegraphics[width=1\linewidth]{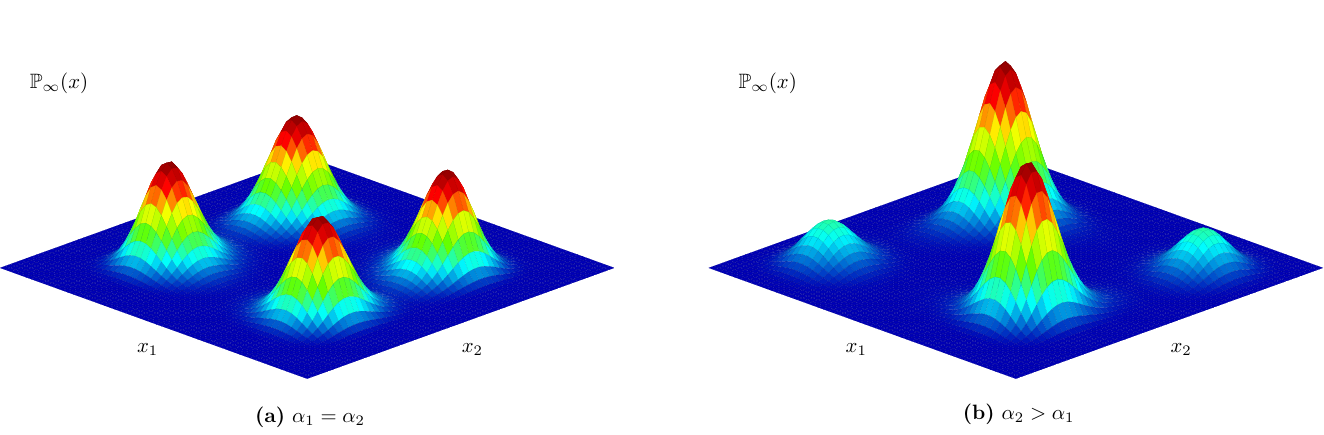}
\caption{Stationary probability densities $\mathbb{P}_{\infty}(x)$ (plotted as surface for $N=2$) that are equilibrium solutions of the F.P.E.~\eqref{FPE}. 
(a): stationary probability density when two memories have equal saliency weights ($\alpha_1=\alpha_2=3/2$). In this case, the IDP Hopfield model is equivalent to the classic Hopfield model. 
(b): stationary probability density when a single memory is dominant ($\alpha_2=9/4>\alpha_1=3/2$). In this case, the memory associated to the dominant saliency weight absorbs most of the probability density and thus enforces a stochastic confinement of the  trajectories.}
\label{fig:PDFs}
\end{figure}
The preliminary study of the stationary probability densities associated to both the classic and IDP Hopfield models provides only a partial account of their differences during the retrieval task. Since the retrieval of memories is a dynamical process, we are also interested in understanding how the noise affects the transients between memories as the input changes. A first observation is that as the amplitude of the perturbations grows, the classic and the IDP Hopfield model increasingly diverge in terms of their response to the input. Indeed, the classic Hopfield model is characterized by a transient that is perfectly aligned with the provided input, but as soon as this ceases to exist, the noise disrupts the network activity and hinders it from reaching any prototypical memory. In contrast to such retrieval breaking, the IDP model still displays a solid capability of correctly selecting the memory associated to the dominant $\alpha$-weight and confining the stochastic wandering to its basin of attraction (see Fig.~S5 in the SI for a comparison of the classic and the new model as they react to mixed inputs and noise). 
  Thus, the stochastic IDP model embodies a mechanism that actively
  prioritizes the most relevant features in its environment while
  suppressing the processing of background elements, mirroring the
  psychological phenomenon of selective attention
  \citep{driver2001selective}.

In addition, the introduction of noise reveals another intriguing property of the IDP model: fast switching in highly uncertain contexts. As can be observed from Fig.~\ref{fig:det_comp}(d), the times elapsing between input onset and correct retrieval in the case of zero noise and instability of the previous memory are quite long. Despite the successful switching, the prolonged dwelling on the previous memory seems somewhat undesirable for systems that have evolved to efficiently and readily adapt to their environment. Moreover, as observable from Fig.~\ref{fig:noise}(b), in cases with zero noise and persistent stability of the previous memory, the system remains entrapped in a previous recall, despite the shift in the dominant input component. This behavior seems equally undesirable, as it would hinder the system from correctly responding to contexts where the incoming information is highly mixed and therefore sufficiently vague.  The IDP model embedded in a natural noisy environment efficiently solves both of these problems, as clear from Fig.~\ref{fig:noise}(c). First, it induces correct switching between different memories in accordance with the current dominant component in the input. Additionally, it does so with very swift transients, presenting a system that is suitably tuned for fast tracking of the external input.

Nonetheless, notice how the IDP model displays, after each new input onset, a relatively brief time period (approximately $t=1$ of simulation time) in which the network activity still remains fixated on the previous recall, and only later adapts to the incoming present information. Short fixation may prove useful to correct glitchy inputs decomposition, during which the system is briefly tricked into decoding a change in the focus of attention (see Fig.~S6 in the SI).

\begin{SCfigure*}[\sidecaptionrelwidth][tbph!]
\centering
\includegraphics[width=.5\linewidth]{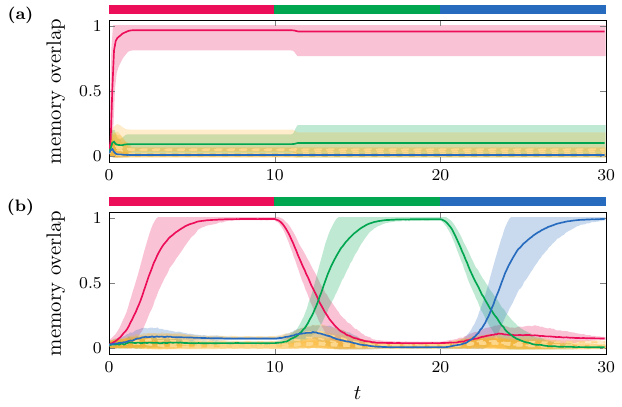}
\caption{Noise facilitates rapid transitions between different prototypical memories for highly mixed inputs in the IDP Hopfield model. Differently from Fig.~\ref{fig:det_comp}, the prototypical memory associated to the red weight remains stable after each input switching. (a) IDP model without noise. The model correctly retrieves the first prototypical memory, but is unable to track the dominant component of new inputs due to persistent stability of the red component. (b) IDP model with noise. The model correctly retrieves the first prototypical memory. At input switch, the memory remains stable, but the associated energy well significantly flattens (see Fig.~\ref{fig:EnergyCases}(d) as reference case). The noise pushes the activity out of the flat wells and directs it towards the deepest one, which is associated to the new dominant saliency weight. The process is then repeated at each input switch.}
\label{fig:noise}
\end{SCfigure*}

\section{Discussion}

\subsection*{Achieving Robust Behavior in Multistable Systems}
Multistability is a distinctive feature of the classical Hopfield model, required for its associative memory functionality. Yet, multistability is typically avoided in controlled engineering systems because it potentially leads to fragile and unpredictable behavior.  While complex systems, such as the electric power grid, may in fact exhibit multiple stability regions, only one of them corresponds to proper grid functioning. It is the responsibility of the grid control infrastructure to steer the system away from undesired stability regions.  The proposed IDP Hopfield model fulfills this same responsibility in the context of memory retrieval: when the
input is not ambiguous, the neural state is reliably driven to the
correct stability region, achieving a robust and predictable memory system.

\subsection*{Energy Shaping in the IDP Hopfield Model} 
The IDP Hopfield model presents a simple yet effective explanation of how a direct input-driven modulation of the synapses can enrich the dynamic range of recurrent neural networks. The input driven adjustments of the synaptic couplings between neurons enforce a clear memory hierarchy, with single memories existing only if sufficiently stimulated. Furthermore, the input decomposition changes the stability properties of single memory patterns. Through the existence of the threshold $\subscr{\alpha}{stability}$, a memory that was stable at a certain time instant can suddenly become a saddle point, and thus allow the network dynamics to roll towards another memory, as proposed by Karuvally et al.~in \citep{karuvally2023general}. During the retrieval process, the saliency attributed to individual memories via input decomposition reshapes the energy landscape of the model. This process deepens the wells associated with the most dominant input components while flattening the others. Consequently, the presence of exogenous noise is enough to drive the dynamics towards the deepest basin of attraction and confine the network activity within it. It is worth mentioning that a model similar to the IDP Hopfield has been recently numerically studied in ~\citep{Herron2023} with the aim of implementing sequential memory retrieval. In this context, the dynamics and distribution of the saliency weights reflect some previous association among prototypical memories, and confine the network activity to limit cycles.

\subsection*{Future Work and Implications}
The present study lays the foundation for future
research aimed at fully analyzing the biologically plausible firing rate
version of the IDP model, a preliminary version of which is outlined in
the SI. This study also opens pathways for the empirical validation of
input-dependent associative memory models and a comprehensive
characterization of learning dynamics, encompassing both long-term memory
formation and short-term modulation of saliency weights. Specifically,
future works will need to consider a complete theoretical investigation of
the IDP firing rate model in a low neural activity regime compatible with
experimental evidence \citep{gastaldi2021shared}. A thorough understanding
of the dynamics underlying memory retrieval in the firing rate model may in
fact guide experimental researchers in their inquiries, and successful
validation of the model may provide a solid ground for the understanding of
high level cognitive processes. Most importantly, our formalization of the IDP Hopfield model paves the way for a thorough investigation into its underlying learning dynamics.

In addition, jointly solving the problem of long-term learning and short-term modulations may provide valuable answers on the multiplicity of timescales observed throughout biological learning processes. Notably, there is already ample evidence \citep{benna2016computational} that heterogeneity in functional features and timescales may be the biological solution to the problems of catastrophic forgetting and adaptive learning. The advancement of such comprehension holds significant implications, extending its impact beyond neuroscience and onto the domain of machine learning (ML). Within the ML field, the mentioned problems prevent virtual and embedded agents \citep{hadsell2020embracing} from deployment in contexts that require both cognitive flexibility and preservation of past knowledge. The successful combination of these two elements epitomizes the ultimate aspiration of continual learning \citep{lesort2020continual,kudithipudi2022biological}, an emergent ML paradigm that seeks to meaningfully integrate present and past data and yields the potential for great technological advancement.

\subsection*{Data, Materials, and Software Availability}
There are no data underlying this work. The code used in this manuscript is publicly available at \url{https://github.com/sim1bet/Input-Driven-Hopfield}.

\subsection*{Acknowledgments}
This work was in part supported by AFOSR project FA9550-21-1-0203 and by the grant Next Generation EU C96E22000350007.

\bibliographystyle{unsrtnat}
\bibliography{references_Main}

\begin{thebibliography}{43}
\providecommand{\natexlab}[1]{#1}
\providecommand{\url}[1]{\texttt{#1}}
\expandafter\ifx\csname urlstyle\endcsname\relax
  \providecommand{\doi}[1]{doi: #1}\else
  \providecommand{\doi}{doi: \begingroup \urlstyle{rm}\Url}\fi

\bibitem[Hopfield(1982)]{hopfield1982neural}
J.~J. Hopfield.
\newblock Neural networks and physical systems with emergent collective
  computational abilities.
\newblock \emph{Proceedings of the National Academy of Sciences}, 79\penalty0
  (8):\penalty0 2554–2558, April 1982.
\newblock ISSN 1091-6490.
\newblock \doi{10.1073/pnas.79.8.2554}.
\newblock URL \url{http://dx.doi.org/10.1073/pnas.79.8.2554}.

\bibitem[Hopfield(1984)]{hopfield1984neurons}
J.~J. Hopfield.
\newblock Neurons with graded response have collective computational properties
  like those of two-state neurons.
\newblock \emph{Proceedings of the National Academy of Sciences}, 81\penalty0
  (10):\penalty0 3088–3092, May 1984.
\newblock ISSN 1091-6490.
\newblock \doi{10.1073/pnas.81.10.3088}.
\newblock URL \url{http://dx.doi.org/10.1073/pnas.81.10.3088}.

\bibitem[Amit et~al.(1987{\natexlab{a}})Amit, Gutfreund, and
  Sompolinsky]{amit1987Stat}
D.~J. Amit, H.~Gutfreund, and H.~Sompolinsky.
\newblock Statistical mechanics of neural networks near saturation.
\newblock \emph{Annals of Physics}, 173\penalty0 (1):\penalty0 30–67, January
  1987{\natexlab{a}}.
\newblock ISSN 0003-4916.
\newblock \doi{10.1016/0003-4916(87)90092-3}.
\newblock URL \url{http://dx.doi.org/10.1016/0003-4916(87)90092-3}.

\bibitem[Crisanti et~al.(1986)Crisanti, Amit, , and
  Gutfreund]{crisanti1986saturation}
A.~Crisanti, D.~J. Amit, , and H.~Gutfreund.
\newblock Saturation level of the {Hopfield} model for neural network.
\newblock \emph{Europhysics Letters (EPL)}, 2\penalty0 (4):\penalty0 337–341,
  August 1986.
\newblock ISSN 1286-4854.
\newblock \doi{10.1209/0295-5075/2/4/012}.
\newblock URL \url{http://dx.doi.org/10.1209/0295-5075/2/4/012}.

\bibitem[Treves and Amit(1988)]{treves1988metastable}
A.~Treves and D.J. Amit.
\newblock Metastable states in asymmetrically diluted {Hopfield} networks.
\newblock \emph{Journal of Physics A: Mathematical and General}, 21\penalty0
  (14):\penalty0 3155–3169, July 1988.
\newblock ISSN 1361-6447.
\newblock \doi{10.1088/0305-4470/21/14/016}.
\newblock URL \url{http://dx.doi.org/10.1088/0305-4470/21/14/016}.

\bibitem[Krotov and Hopfield(2016)]{krotov2016dense}
D.~Krotov and J.J. Hopfield.
\newblock Dense associative memory for pattern recognition.
\newblock \emph{Advances in neural information processing systems}, 29, 2016.
\newblock URL \url{https://arxiv.org/pdf/1606.01164}.

\bibitem[Krotov and Hopfield(2020)]{krotov2020large}
D.~Krotov and J.J. Hopfield.
\newblock Large associative memory problem in neurobiology and machine
  learning.
\newblock In \emph{International Conference on Learning Representations}, 2020.
\newblock URL \url{https://arxiv.org/pdf/2008.06996}.

\bibitem[Chaudhry et~al.(2023)Chaudhry, Zavatone-Veth, Krotov, and
  Pehlevan]{chaudhry2023long}
Hamza~Tahir Chaudhry, Jacob~A Zavatone-Veth, Dmitry Krotov, and Cengiz
  Pehlevan.
\newblock Long sequence {Hopfield} memory.
\newblock In \emph{Associative Memory {\&} Hopfield Networks in 2023}, 2023.
\newblock URL \url{https://openreview.net/forum?id=sYAm62gWbo}.

\bibitem[Demircigil et~al.(2017)Demircigil, Heusel, Lowe, Upgang, , and
  Vermet]{demircigil2017model}
M.~Demircigil, J.~Heusel, M.~Lowe, S.~Upgang, , and F.~Vermet.
\newblock On a model of associative memory with huge storage capacity.
\newblock \emph{Journal of Statistical Physics}, 168\penalty0 (2):\penalty0
  288–299, May 2017.
\newblock ISSN 1572-9613.
\newblock \doi{10.1007/s10955-017-1806-y}.
\newblock URL \url{http://dx.doi.org/10.1007/s10955-017-1806-y}.

\bibitem[Hubert et~al.(2021)Hubert, Bernhard, Johannes, Philipp, Michael,
  Lukas, Markus, Thomas, David, Michael, G{\"u}nter, Johannes, and
  Sepp]{ramsauer2021hopfield}
R.~Hubert, S.~Bernhard, L.~Johannes, S.~Philipp, W.~Michael, G.~Lukas,
  H.~Markus, A.~Thomas, K.~David, K.K. Michael, K.~G{\"u}nter, B.~Johannes, and
  H.~Sepp.
\newblock Hopfield networks is all you need.
\newblock In \emph{International Conference on Learning Representations}, 2021.
\newblock URL \url{https://openreview.net/forum?id=tL89RnzIiCd}.

\bibitem[Kozachkov et~al.(2023)Kozachkov, Kastanenka, and
  Krotov]{kozachkov2023building}
L.~Kozachkov, K.~V. Kastanenka, and D.~Krotov.
\newblock Building transformers from neurons and astrocytes.
\newblock \emph{Proceedings of the National Academy of Sciences}, 120\penalty0
  (34), August 2023.
\newblock ISSN 1091-6490.
\newblock \doi{10.1073/pnas.2219150120}.
\newblock URL \url{http://dx.doi.org/10.1073/pnas.2219150120}.

\bibitem[Treves and Rolls(1992)]{treves1992computational}
Alessandro Treves and Edmund~T. Rolls.
\newblock Computational constraints suggest the need for two distinct input
  systems to the hippocampal {CA3} network.
\newblock \emph{Hippocampus}, 2\penalty0 (2):\penalty0 189–199, April 1992.
\newblock ISSN 1098-1063.
\newblock \doi{10.1002/hipo.450020209}.
\newblock URL \url{http://dx.doi.org/10.1002/hipo.450020209}.

\bibitem[Yassa and Stark(2011)]{yassa2011pattern}
M.A. Yassa and Craig~E.L. Stark.
\newblock Pattern separation in the hippocampus.
\newblock \emph{Trends in Neurosciences}, 34\penalty0 (10):\penalty0 515–525,
  October 2011.
\newblock ISSN 0166-2236.
\newblock \doi{10.1016/j.tins.2011.06.006}.
\newblock URL \url{http://dx.doi.org/10.1016/j.tins.2011.06.006}.

\bibitem[Rolls(2013)]{rolls2013mechanisms}
E.T. Rolls.
\newblock The mechanisms for pattern completion and pattern separation in the
  hippocampus.
\newblock \emph{Frontiers in Systems Neuroscience}, 7, 2013.
\newblock ISSN 1662-5137.
\newblock \doi{10.3389/fnsys.2013.00074}.
\newblock URL \url{http://dx.doi.org/10.3389/fnsys.2013.00074}.

\bibitem[Russo and Treves(2012)]{russo2012cortical}
E.~Russo and A.~Treves.
\newblock Cortical free-association dynamics: Distinct phases of a latching
  network.
\newblock \emph{Physical Review E}, 85\penalty0 (5), May 2012.
\newblock ISSN 1550-2376.
\newblock \doi{10.1103/physreve.85.051920}.
\newblock URL \url{http://dx.doi.org/10.1103/PhysRevE.85.051920}.

\bibitem[Naim et~al.(2018)Naim, Boboeva, Kang, and Treves]{naim2018reducing}
M.~Naim, V.~Boboeva, C.J. Kang, and A.~Treves.
\newblock Reducing a cortical network to a {Potts} model yields storage
  capacity estimates.
\newblock \emph{Journal of Statistical Mechanics: Theory and Experiment},
  2018\penalty0 (4):\penalty0 043304, April 2018.
\newblock ISSN 1742-5468.
\newblock \doi{10.1088/1742-5468/aab683}.
\newblock URL \url{http://dx.doi.org/10.1088/1742-5468/aab683}.

\bibitem[Amit(1989)]{amit1989modeling}
D.~J. Amit.
\newblock \emph{Modeling Brain Function: The World of Attractor Neural
  Networks}.
\newblock Cambridge University Press, September 1989.
\newblock ISBN 9780511623257.
\newblock \doi{10.1017/cbo9780511623257}.
\newblock URL \url{http://dx.doi.org/10.1017/CBO9780511623257}.

\bibitem[Dayan and Abbott(2005)]{dayan2005theoretical}
P.~Dayan and L.~F. Abbott.
\newblock \emph{Theoretical Neuroscience: Computational and Mathematical
  Modeling of Neural Systems}.
\newblock MIT Press, 2005.

\bibitem[Gerstner et~al.(2014)Gerstner, Kistler, Naud, and
  Paninski]{gerstner2014neuronal}
W.~Gerstner, W.~M. Kistler, R.~Naud, and L.~Paninski.
\newblock \emph{Neuronal Dynamics: From Single Neurons to Networks and Models
  of Cognition}.
\newblock Cambridge University Press, July 2014.
\newblock ISBN 9781107447615.
\newblock \doi{10.1017/cbo9781107447615}.
\newblock URL \url{http://dx.doi.org/10.1017/CBO9781107447615}.

\bibitem[Zenke et~al.(2017)Zenke, Poole, and Ganguli]{zenke2017continual}
F.~Zenke, B.~Poole, and S.~Ganguli.
\newblock Continual learning through synaptic intelligence.
\newblock In \emph{International Conference on Machine Learning}, pages
  3987--3995. PMLR, 2017.

\bibitem[Hadsell et~al.(2020)Hadsell, Rao, Rusu, and
  Pascanu]{hadsell2020embracing}
R.~Hadsell, D.~Rao, A.~A. Rusu, and R.~Pascanu.
\newblock Embracing change: Continual learning in deep neural networks.
\newblock \emph{Trends in Cognitive Sciences}, 24\penalty0 (12):\penalty0
  1028–1040, December 2020.
\newblock ISSN 1364-6613.
\newblock \doi{10.1016/j.tics.2020.09.004}.
\newblock URL \url{http://dx.doi.org/10.1016/j.tics.2020.09.004}.

\bibitem[Lesort et~al.(2020)Lesort, Lomonaco, Stoian, Maltoni, Filliat, and
  Díaz-Rodríguez]{lesort2020continual}
T.~Lesort, V.~Lomonaco, A.~Stoian, D.~Maltoni, D.~Filliat, and
  N.~Díaz-Rodríguez.
\newblock Continual learning for robotics: Definition, framework, learning
  strategies, opportunities and challenges.
\newblock \emph{Information Fusion}, 58:\penalty0 52–68, June 2020.
\newblock ISSN 1566-2535.
\newblock \doi{10.1016/j.inffus.2019.12.004}.
\newblock URL \url{http://dx.doi.org/10.1016/j.inffus.2019.12.004}.

\bibitem[Amit et~al.(1987{\natexlab{b}})Amit, Gutfreund, , and
  Sompolinsky]{Amit1987Inf}
D.~J. Amit, H.~Gutfreund, , and H.~Sompolinsky.
\newblock Information storage in neural networks with low levels of activity.
\newblock \emph{Physical Review A}, 35\penalty0 (5):\penalty0 2293–2303,
  March 1987{\natexlab{b}}.
\newblock ISSN 0556-2791.
\newblock \doi{10.1103/physreva.35.2293}.
\newblock URL \url{http://dx.doi.org/10.1103/PhysRevA.35.2293}.

\bibitem[Petritis(1996)]{petritis1995thermodynamic}
D.~Petritis.
\newblock \emph{Thermodynamic Formalism of Neural Computing}, page 81–146.
\newblock Springer, 1996.
\newblock ISBN 9789401713238.
\newblock \doi{10.1007/978-94-017-1323-8_3}.
\newblock URL \url{http://dx.doi.org/10.1007/978-94-017-1323-8_3}.

\bibitem[Battaglia and Treves(1998)]{Battaglia1998}
Francesco~P. Battaglia and Alessandro Treves.
\newblock Stable and rapid recurrent processing in realistic autoassociative
  memories.
\newblock \emph{Neural Computation}, 10\penalty0 (2):\penalty0 431–450,
  February 1998.
\newblock ISSN 1530-888X.
\newblock \doi{10.1162/089976698300017827}.
\newblock URL \url{http://dx.doi.org/10.1162/089976698300017827}.

\bibitem[Blumenfeld et~al.(2006)Blumenfeld, Preminger, Sagi, and
  Tsodyks]{blumenfeld2006dynamics}
B.~Blumenfeld, S.~Preminger, D.~Sagi, and M.~Tsodyks.
\newblock Dynamics of memory representations in networks with
  novelty-facilitated synaptic plasticity.
\newblock \emph{Neuron}, 52\penalty0 (2):\penalty0 383–394, October 2006.
\newblock ISSN 0896-6273.
\newblock \doi{10.1016/j.neuron.2006.08.016}.
\newblock URL \url{http://dx.doi.org/10.1016/j.neuron.2006.08.016}.

\bibitem[Tang et~al.(2010)Tang, Li, and Yan]{tang2010memory}
Huajin Tang, Haizhou Li, and Rui Yan.
\newblock Memory dynamics in attractor networks with saliency weights.
\newblock \emph{Neural Computation}, 22\penalty0 (7):\penalty0 1899–1926,
  July 2010.
\newblock ISSN 1530-888X.
\newblock \doi{10.1162/neco.2010.07-09-1050}.
\newblock URL \url{http://dx.doi.org/10.1162/neco.2010.07-09-1050}.

\bibitem[Mongillo et~al.(2008)Mongillo, Barak, and
  Tsodyks]{mongillo2008synaptic}
G.~Mongillo, O.~Barak, and M.~Tsodyks.
\newblock Synaptic theory of working memory.
\newblock \emph{Science}, 319\penalty0 (5869):\penalty0 1543–1546, March
  2008.
\newblock ISSN 1095-9203.
\newblock \doi{10.1126/science.1150769}.
\newblock URL \url{http://dx.doi.org/10.1126/science.1150769}.

\bibitem[Yan et~al.(2013)Yan, Zhao, Hu, Wang, Wang, and
  Wang]{yan2013nonequilibrium}
H.~Yan, L.~Zhao, L.~Hu, X.~Wang, E.~Wang, and J.~Wang.
\newblock Nonequilibrium landscape theory of neural networks.
\newblock \emph{Proceedings of the National Academy of Sciences}, 110\penalty0
  (45), October 2013.
\newblock ISSN 1091-6490.
\newblock \doi{10.1073/pnas.1310692110}.
\newblock URL \url{http://dx.doi.org/10.1073/pnas.1310692110}.

\bibitem[Brinkman et~al.(2022)Brinkman, Yan, Maffei, Park, Fontanini, Wang, and
  La~Camera]{brinkman2022metastable}
B.~A.~W. Brinkman, H.~Yan, A.~Maffei, I.~M. Park, A.~Fontanini, J.~Wang, and
  G.~La~Camera.
\newblock Metastable dynamics of neural circuits and networks.
\newblock \emph{Applied Physics Reviews}, 9\penalty0 (1), March 2022.
\newblock ISSN 1931-9401.
\newblock \doi{10.1063/5.0062603}.
\newblock URL \url{http://dx.doi.org/10.1063/5.0062603}.

\bibitem[Moss and McClintock(1989)]{moss1989noise}
F.~Moss and P.V.E. McClintock.
\newblock \emph{Noise in Nonlinear Dynamical Systems}.
\newblock Cambridge University Press, April 1989.
\newblock ISBN 9780511897818.
\newblock \doi{10.1017/cbo9780511897818}.
\newblock URL \url{http://dx.doi.org/10.1017/CBO9780511897818}.

\bibitem[Wong(1991)]{wong1991stochastic}
E.~Wong.
\newblock Stochastic neural networks.
\newblock \emph{Algorithmica}, 6\penalty0 (1–6):\penalty0 466–478, June
  1991.
\newblock ISSN 1432-0541.
\newblock \doi{10.1007/bf01759054}.
\newblock URL \url{http://dx.doi.org/10.1007/BF01759054}.

\bibitem[Driver(2001)]{driver2001selective}
J.~Driver.
\newblock A selective review of selective attention research from the past
  century.
\newblock \emph{British Journal of Psychology}, 92\penalty0 (1):\penalty0
  53–78, February 2001.
\newblock ISSN 2044-8295.
\newblock \doi{10.1348/000712601162103}.
\newblock URL \url{http://dx.doi.org/10.1348/000712601162103}.

\bibitem[Karuvally et~al.(2023)Karuvally, Sejnowski, and
  Siegelmann]{karuvally2023general}
Arjun Karuvally, Terrence Sejnowski, and Hava~T Siegelmann.
\newblock General sequential episodic memory model.
\newblock In \emph{International Conference on Machine Learning}, pages
  15900--15910. PMLR, 2023.

\bibitem[Herron et~al.(2023)Herron, Sartori, and Xue]{Herron2023}
Lukas Herron, Pablo Sartori, and BingKan Xue.
\newblock Robust retrieval of dynamic sequences through interaction modulation.
\newblock \emph{PRX Life}, 1\penalty0 (2), December 2023.
\newblock ISSN 2835-8279.
\newblock \doi{10.1103/prxlife.1.023012}.
\newblock URL \url{http://dx.doi.org/10.1103/PRXLife.1.023012}.

\bibitem[Gastaldi et~al.(2021)Gastaldi, Schwalger, De~Falco, Quiroga, and
  Gerstner]{gastaldi2021shared}
C.~Gastaldi, T.~Schwalger, E.~De~Falco, R.~Q. Quiroga, and W.~Gerstner.
\newblock When shared concept cells support associations: Theory of overlapping
  memory engrams.
\newblock \emph{PLOS Computational Biology}, 17\penalty0 (12):\penalty0
  e1009691, December 2021.
\newblock ISSN 1553-7358.
\newblock \doi{10.1371/journal.pcbi.1009691}.
\newblock URL \url{http://dx.doi.org/10.1371/journal.pcbi.1009691}.

\bibitem[Benna and Fusi(2016)]{benna2016computational}
M.~K. Benna and S.~Fusi.
\newblock Computational principles of synaptic memory consolidation.
\newblock \emph{Nature Neuroscience}, 19\penalty0 (12):\penalty0 1697–1706,
  October 2016.
\newblock ISSN 1546-1726.
\newblock \doi{10.1038/nn.4401}.
\newblock URL \url{http://dx.doi.org/10.1038/nn.4401}.

\bibitem[Kudithipudi et~al.(2022)Kudithipudi, Aguilar-Simon, Babb, Bazhenov,
  Blackiston, Bongard, Brna, Chakravarthi~Raja, Cheney, Clune, Daram, Fusi,
  Helfer, Kay, Ketz, Kira, Kolouri, Krichmar, Kriegman, Levin, Madireddy,
  Manicka, Marjaninejad, McNaughton, Miikkulainen, Navratilova, Pandit, Parker,
  Pilly, Risi, Sejnowski, Soltoggio, Soures, Tolias, Urbina-Meléndez,
  Valero-Cuevas, van~de Ven, Vogelstein, Wang, Weiss, Yanguas-Gil, Zou, and
  Siegelmann]{kudithipudi2022biological}
D.~Kudithipudi, M.~Aguilar-Simon, J.~Babb, M.~Bazhenov, D.~Blackiston,
  J.~Bongard, A.P. Brna, S.~Chakravarthi~Raja, N.~Cheney, J.~Clune, A.~Daram,
  S.~Fusi, P.~Helfer, L.~Kay, N.~Ketz, Z.~Kira, S.~Kolouri, J.L. Krichmar,
  S.~Kriegman, M.~Levin, S.~Madireddy, S.~Manicka, A.~Marjaninejad,
  B.~McNaughton, R.~Miikkulainen, Z.~Navratilova, T.~Pandit, A.~Parker, P.K.
  Pilly, S.~Risi, T.J. Sejnowski, A.~Soltoggio, N.~Soures, A.S. Tolias,
  D.~Urbina-Meléndez, F.J. Valero-Cuevas, G.M. van~de Ven, J.T. Vogelstein,
  F.~Wang, R.~Weiss, A.~Yanguas-Gil, X.~Zou, and H.~Siegelmann.
\newblock Biological underpinnings for lifelong learning machines.
\newblock \emph{Nature Machine Intelligence}, 4\penalty0 (3):\penalty0
  196–210, March 2022.
\newblock ISSN 2522-5839.
\newblock \doi{10.1038/s42256-022-00452-0}.
\newblock URL \url{http://dx.doi.org/10.1038/s42256-022-00452-0}.

\bibitem[Gelbrich and R\"{o}misch(1995)]{Gelbrich1995}
Matthias Gelbrich and Werner R\"{o}misch.
\newblock Numerical solution of stochastic differential equations (peter e.
  kloeden and eckhard platen).
\newblock \emph{SIAM Review}, 37\penalty0 (2):\penalty0 272–275, June 1995.
\newblock ISSN 1095-7200.
\newblock \doi{10.1137/1037073}.
\newblock URL \url{http://dx.doi.org/10.1137/1037073}.

\bibitem[Horn and Weyers(1987)]{Horn1987}
D.~Horn and J.~Weyers.
\newblock Hypercubic structures in orthogonal {Hopfield} models.
\newblock \emph{Physical Review A}, 36\penalty0 (10):\penalty0 4968–4974,
  November 1987.
\newblock ISSN 0556-2791.
\newblock \doi{10.1103/physreva.36.4968}.
\newblock URL \url{http://dx.doi.org/10.1103/PhysRevA.36.4968}.

\bibitem[Horn and Johnson(1985)]{Horn1985}
Roger~A. Horn and Charles~R. Johnson.
\newblock \emph{Matrix Analysis}.
\newblock Cambridge University Press, December 1985.
\newblock ISBN 9780511810817.
\newblock \doi{10.1017/cbo9780511810817}.
\newblock URL \url{http://dx.doi.org/10.1017/CBO9780511810817}.

\bibitem[Khalil(2002)]{Khalil2002}
H.K. Khalil.
\newblock \emph{Nonlinear Systems}.
\newblock Prentice Hall, 2002.

\bibitem[Blanchini and Miani(2015)]{Blanchini2015}
Franco Blanchini and Stefano Miani.
\newblock \emph{Set-Theoretic Methods in Control}.
\newblock Springer International Publishing, 2015.
\newblock ISBN 9783319179339.
\newblock \doi{10.1007/978-3-319-17933-9}.
\newblock URL \url{http://dx.doi.org/10.1007/978-3-319-17933-9}.

\end{thebibliography}

\newpage
\section*{\Huge Supporting Information}

\section*{\Large Notation}

We let $\real^{n\times m}$ denote the set of $n\times m$ matrices with real entries. The symbol $\vectorones[n]$ indicates an $n$-dimensional vectors of ones, $\vectorzeros[n]$ an $n$-dimensional vector of zeros, and $I_n$ the $n\times n$ identity matrix (we will drop the subscript when the dimension is clear from the context).  Given a matrix $A\in\real^{n\times m}$, $A^\top$ is the transpose of $A$ and $\|A\|$ the 2-norm of $A$. For two generic vectors $x,y\in\real^{n}$, we denote with $\langle{x,y}\rangle:=x^\top y$ their inner product. For a symmetric matrix $A=A^\top$, we let $\subscr{\lambda}{max}(A)$ denote the largest eigenvalue of $A$ and we write $A\succ 0$ ($A\succeq 0$) if $A$ is positive definite (positive semidefinite, respectively). We denote with $\mathrm{diag}(d_1,\dots,d_n)$ the diagonal matrix with diagonal entries $d_1,\dots, d_n$. We denote with $C^{k}(\real^{N};\real^{M})$ the set of continuously k-differentiable functions with domain $\real^{N}$ and co-domain $\real^{M}$. If $f(x)\in C^{k}(\real; \real)$, $f^{(k)}(x)$ denotes the $k$-th derivative of $f$. Let $f\in{C}^{1}(\real^{N};\real)$, and we denote with $\nabla_{x}$ the gradient operator such that $\nabla_{x}f(x)=(\partial_{x_{1}}f(x),\dots,\partial_{x_{N}}f(x))$, where $\partial_{x_{i}}$ is the partial derivative w.r.t. the $x_{i}$ variable. We let $\mathbb{P}[\mathcal{A}]$ be probability of the event $\mathcal{A}$, and for a generic random variable $\mathcal{X}$ we denote its expected value as $\mathbb{E}[\mathcal{X}]$. Finally, we denote with $\delta(x)$ the Dirac delta centered at $x=0$.

\section*{Materials and Methods}

\subsection*{Characterization of the input}
For all our simulations, in order to have experimental control over the
interpretation of the resulting trajectories, we define the inputs as:
\begin{equation}
    u = \sum_{\mu=1}^{P}\alpha_{\mu}\xi^{\mu}
\end{equation}
where the saliency weights $\alpha_{\mu}$, $\mu=1,\dots,P$ were randomly
drawn from a uniform distribution. Specifically, fix $\nu\in\{1,\dots,P\}$
to be the dominant saliency weight. Then we have that
$\alpha_{\nu}\sim\mathcal{U}([2,3.5])$ and
$\alpha_{\mu}\sim\mathcal{U}([0.8,1.5])$ for all $\mu\neq\nu$,
$\mu=1,\dots,P$.  Furthermore, to prevent any memory $\xi^{\mu}$ from being
an equilibrium point for the system, it suffices to set $\alpha_{\mu}<0.2$
given our choice of the slope for the activation function.

\subsection*{Simulation of the Trajectories}
Each simulation of the trajectories starts from an initial condition
sampled as $x_{0}\sim\mathcal{N}(0,\mathrm{I}_{N})$. The evolution of the
trajectory is then computed using the Euler-Maruyama method
\citep{Gelbrich1995} as
\begin{equation}
    x(t+\delta t) = x(t) + \mathcal{F}(x(t))\delta t + \sigma\bar \eta\sqrt{\delta t} 
\end{equation}
where $\delta t = 0.01$ is the step-size for the integration, $\sigma=8$ is
the standard deviation of the noise, and $\bar
\eta\sim\mathcal{N}(0,\mathrm{I}_{N})$. For the deterministic simulation,
we simply take $\sigma=0$ and recover the standard forward Euler method.

For the sample average of the trajectories in the main text, the
deterministic simulations were averaged over 50 different initial
conditions. Instead, for the stochastic simulations the average was taken
over 50 different initial conditions and 50 different noise realizations.

\section*{\Large The Model and Existence of Equilibria}
\subsection*{The Model}
Define a set of orthogonal prototypical memory vectors (which can be relaxed to i.i.d.\ memories in the limit of a large number of neurons)
\begin{align}
    &\Sigma :=\{\xi^{1},\dots,\xi^{p}\}, \ \ \xi^{\mu} \in \{-1,1\}^{N} \label{eq:prot-mem1}\\
    & {\xi^{\mu}}^{\top}\xi^{\nu} = \begin{cases}N & \text{if } \mu=\nu \\ 0 & \text{if } \mu\neq \nu\end{cases} \label{eq:prot-mem2}
\end{align}
where $N\in{\mathbb{N}}$ is the number of units of the recurrent neural network. Consider the following dynamical system
\begin{equation}
    \begin{cases}\label{HopU}
        \dot{x}(t) = -x(t) + W(u)\Psi(x(t))\\
        x(0)=x_{0}\in{\real}^{N}
    \end{cases}
\end{equation}
where we define the input modulated synaptic matrix
\begin{equation}
    W(u) = \frac{1}{N}\sum_{\mu=1}^{P}\xi^{\mu}{\xi^{\mu}}^{\top}u(t){\xi^{\mu}}^{\top}
\end{equation}
for a generic constant input $u\in\real^{N}$, where the inner products result in the following saliency weights
\begin{equation}
    \alpha_{\mu}={\xi^{\mu}}^{\top}u(t)\qquad\forall{\mu=1,\dots,P}
\end{equation}

We assume that the activation function satisfies the following conditions:
        \begin{align}\label{eq:diag}
            &\Psi(x) = (\psi(x_{1}),\dots,\psi(x_{N})), \ \ \psi\in{C}^{2}(\real;\real)\\ \label{eq: odd}
            &\psi(z) =-\psi(z)\quad{z\in\real}\\ \label{eq: der}
            &0<\psi'(z)\leq 1,\ \  {z\in{\real}}, \ \ \psi'(0)=1\\ \label{eq: asy}
            &\!\!\lim_{z\rightarrow{\pm{\infty}}}\psi(z) = \pm{1} \\ \label{eq: 2der}
            &\psi''(z)
        \begin{cases}
            <0\qquad{z>0}\\
            =0\qquad{z=0}\\
            >0\qquad{z<0}
        \end{cases}
        \end{align}
    
A function satisfying the above conditions is, for instance, $\psi(z)=\tanh(z)$. Condition~\eqref{eq: der} is assumed for simplicity, and the following analysis can be easily adapted to the case where $0<\psi'(z)\leq \beta$,  ${z\in{\real}}$, $\psi'(0)=\beta$, and $\beta >0$.  

\subsection*{Existence of Equilibria}
Given a prototypical memory  $\xi^{\rho}$ as in~\eqref{eq:prot-mem1}-\eqref{eq:prot-mem2}, a retrievable memory associated to $\xi^\rho$ is a vector of the form $x_{\rho}=\gamma_{\rho}\xi^{\rho}$ for $\gamma_{\rho}\neq{0}$ which is an equilibrium for~\eqref{HopU}.

\begin{theorem}[Existence of retrievable memories]\label{thm:existence} Let $\xi^{\rho}$ be a prototypical memory. Then $x_{\rho}=\gamma_{\rho}\xi^{\rho}$ is an equilibrium for~\eqref{HopU} for some 
$\gamma_{\rho}\in{\real}\setminus\{0\}$ if and only if $\alpha_{\rho}>1$. In this case, $\gamma_{\rho}$ satisfies 
    \begin{align}
        \frac{\gamma_{\rho}}{\alpha_{\rho}}&=\psi(\gamma_{\rho})
    \end{align}   
\end{theorem}
\begin{proof}
    The vector $x_{\rho}=\gamma_{\rho}\xi^{\rho}$ is an equilibrium for~\eqref{HopU} if and only if
    \begin{align}
        x_{\rho} &= \gamma_{\rho}\xi^{\rho}\nonumber\\
                &= W(u)\Psi(\gamma_{\rho}\xi^{\rho})\nonumber\\
                &= \frac{1}{N}\sum_{\mu=1}^{P}\alpha_\mu \xi^{\mu}{\xi^{\mu}}^{\top}\cdot{\xi^{\rho}}\psi(\gamma_{\rho})\nonumber\\
                &=\alpha_{\rho}\xi^{\rho}\psi(\gamma_{\rho})
    \end{align}    
    and hence if and only if
    \begin{equation}
        \gamma_{\rho}=\alpha_\rho\psi(\gamma_{\rho})
    \end{equation}
    The theorem is proved if we prove that 
    \begin{equation}
        \exists \gamma>0 \ : \ \gamma=\alpha\psi(\gamma)\quad\Leftrightarrow \quad \alpha>1
    \end{equation}    
    We start from ($\Rightarrow$). Observe that since $\psi''(z)<0$ for all $z>0$ then $\psi'(z)$ is strictly decreasing for $z>0$ and hence $\psi(z)<z$ for $z>0$. This proves that if there exists $\gamma>0$ such that $\gamma=\alpha\psi(\gamma)$, then $\alpha=\gamma/\psi(\gamma)>1$. We prove now ($\Leftarrow$). Consider the function $f(z):=z-\alpha\psi(z)$ and observe that $f(0)=0$, $f'(0)=1-\alpha<0$ and that $\lim_{z\to+\infty}f(z)=+\infty$. These facts imply that there exists $\gamma>0$ such that $f(\gamma)=0$ and this concludes the proof.
\end{proof} 
Fig.~\ref{fig:ActFun} provides a visualization of the existence condition in Theorem \ref{thm:existence}. By definition, retrievable memories need to be the corners of an hypercube \citep{Horn1987}, and we readily see which conditions on the saliency weights $\alpha$'s grant us their existence. If $\alpha\leq{1}$, then in the semipositive interval $\real_{+}$ the dissipation line $\frac{\gamma}{\alpha}$ has only one intersection with the activation function that coincides with the center of the hypercube. Instead, for values $\alpha>1$ we have two intersections, one coinciding with the desired corner of the hypercube.

\begin{figure}[h!]
    \centering
    \includegraphics[width=.7\linewidth]{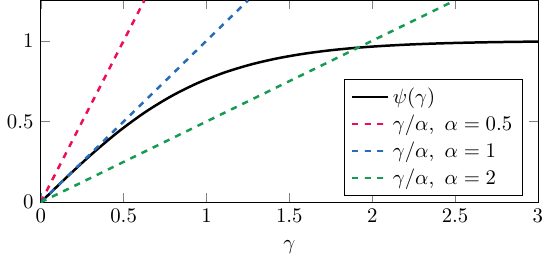}
    \caption{Intersections of the activation function $\psi(\gamma)=\tanh(\gamma)$ with the dissipation line $\frac{\gamma}{\alpha}$ for different values of $\alpha$. The dissipation lines where the parameter $\alpha\leq{1}$ have just one intersection with the chosen activation function. Instead, only for the value $\alpha=2>1$ we obtain two intersections between the dissipation line and the activation function in the semipositive interval $\real_{+}$.}
    \label{fig:ActFun}
\end{figure}

Notice that, since each equilibrium is of the form $x_{\rho}=\gamma_{\rho}\xi^{\rho}$, it suffices to choose a sufficiently fast saturating activation function to obtain that, in the space of membrane potentials, retrievable memories are of the form $\Psi(x_{\rho})=\xi^{\rho}\psi(\gamma_{\rho})\approx{\xi^{\rho}}$.

\section*{Stability of Equilibria and Stability Threshold}

\subsection*{Stability of Equilibria}
In order to determine when the equilibria $x_{\rho}=\gamma_{\rho}\xi^{\rho}$ of our system are stable, we first establish a general local stability condition for Hopfield-type systems. In what follows, given $g\in{C^{1}}(\real^{N};\real^{N})$, we denote the Jacobian of $g(\cdot)$ as the matrix $Dg(x)\in{\real^{N\times{N}}}$ such that ${[Dg(x)]_{ij}=\frac{\partial{g_{i}}}{\partial{x_{j}}}}(x)$. We now specifically address the stability of the IDP Hopfield model. 
\begin{theorem}[Local stability of equilibria and memories]\label{thm: stb-cnd}
Consider the dynamics~\eqref{HopU} with a fixed input $u\in\real^{N}$. Then
\begin{enumerate}
\item an equilibrium $x^{\star}\in{\real}^{N}$ is locally exponentially stable if and only if 
    \begin{equation}
        -[D\Psi(x^{\star})]^{-1}+W(u)\prec{0}
    \end{equation}
    \item if $x_{\rho}=\gamma_{\rho}\xi^{\rho}$ with $\gamma_{\rho}\in{\real}\setminus{\{0\}}$ is an equilibrium, then $x_{\rho}$ is locally exponentially stable if and only if
    \begin{equation}
        {\psi'(\gamma_{\rho})}<\frac{1}{ \max_{\mu=1,\dots,P}\:\alpha_{\mu}}
    \end{equation}
    \end{enumerate}
\end{theorem}
\begin{proof}\ \newline 
\begin{enumerate} 
    \item Consider, for fixed $u\in\real^{N}$, the field associated to the dynamics~\eqref{HopU}
    \begin{equation}
        F(x^{\star})=-x+W(u)\Psi(x)
    \end{equation}
    and compute the corresponding Jacobian at the equilibrium $x^{\star}$ 
    \begin{equation}
        DF(x^{\star})=-I_{N}+W(u){D\Psi(x^{\star})}
    \end{equation}
    Since $DF(x)\in{\real^{N\times{N}}}$ is not necessarily symmetric, we apply the similarity transformation
    \begin{align}
        D\Psi(x^{\star})^{\frac{1}{2}}{DF(x^{\star})}{D\Psi(x^{\star})^{-\frac{1}{2}}}&=D\Psi(x^{\star})^{\frac{1}{2}}\underbrace{\big{(}-[D\Psi(x^{\star})]^{-1}+W(u)\big{)}}_{S(x^{\star})}{D\Psi(x^{\star})^{\frac{1}{2}}}\nonumber
    \end{align}
    where we have exploited the positive definiteness of $D\Psi(x)\in{\real^{N\times{N}}}$. The product
    \begin{equation}
        D\Psi(x^{\star})^{\frac{1}{2}}{S(x^{\star})}D\Psi(x^{\star})^{\frac{1}{2}}
    \end{equation} is a congruence transformation that preserves the matrix inertia \citep{Horn1985}. Now, the above product is symmetric, and in particular, the Jacobian $DF(x^{\star})$ is Hurwitz if and only if
    \begin{equation}\label{eq:S}
        S(x^{\star})=S(x^{\star})^{\top}\prec {0}
    \end{equation}
    Then the thesis of the statement follows from \citep[Theorem 4.15]{Khalil2002}. 
    \item Note that $D\Psi(x_\rho)$  is a diagonal matrix with $i$-th diagonal entry equal to $\psi'\bigl( [x_{\rho}]_i\bigr)$. 
    Since $\xi^{\rho}\in\{-1,+1\}^N$, we know that each entry of $x_{\rho}=\gamma_{\rho}\xi^{\rho}$ is  equal to $\pm \gamma_\rho$. 
    Next, we note that, for all $z\in\real$,
    \begin{align}
        & \psi(z) = -\psi(-z) \ \ \implies \ \  \psi'(-z) =\frac{d}{dz}\big{(}-\psi(-z)\big{)}
        =\frac{d}{dz}\psi(z)=\psi'(z)
    \end{align}
    Therefore, $D\Psi(x_\rho) = \psi'(\gamma_\rho)I_N$ and, in turn,
    the local exponential stability condition~\eqref{eq:S} reads as
    \begin{align}\label{eq:stab-cond}
        \frac{1}{N}\sum_{\mu=1}^{P}\alpha_{\mu}\xi^{\mu}{\xi^{\mu}}^{\top} 
        \prec \frac{1}{\psi'(\gamma_{\rho})}I_{N}
    \end{align}
    Let $\Xi:=\frac{1}{\sqrt{N}}\begin{bmatrix}\xi^1\cdots \xi^P\end{bmatrix}$ and observe that condition~\eqref{eq:stab-cond} can be written as in compact form as
    $$
        \lambda_{\max}(\Xi\,\mathrm{diag}(\alpha_1,\dots,\alpha_P)\Xi^\top)<\frac{1}{\psi'(\gamma_{\rho})}
    $$
    Letting $A:=\Xi\,\mathrm{diag}(\alpha_1,\dots,\alpha_P)^{\frac 1 2}$ and using the identity $\lambda_{\max}(AA^\top)=\lambda_{\max}(A^\top A)$, the left-hand side of the previous inequality equals
    \begin{align}
        \lambda_{\max}\left(A^\top A\right)&= \lambda_{\max}\left(\mathrm{diag}(\alpha_1,\dots,\alpha_P)^{\frac 1 2}\Xi^\top\Xi\,\mathrm{diag}(\alpha_1,\dots,\alpha_P)^{\frac 1 2}\right)\nonumber\\&= \lambda_{\max}(\mathrm{diag}(\alpha_1,\dots,\alpha_P))=\max_{\mu=1,\dots,P}\alpha_\mu
    \end{align}   
    where we used the fact that $\Xi^\top \Xi=I$.
\end{enumerate}
\end{proof}

\subsection*{Stability Threshold}
From the previous stability results, it is natural to wonder whether there exists a critical parameter $\subscr{\alpha}{stability}$ that determines the stability of the memory patterns, and how it is possible to quantify it. In what follows, to lighten the notation, we will refer to $\subscr{\alpha}{stability}$ as $\alpha^{\star}$. 
For simplicity and without loss of generality, we impose the following order on the saliency weights
    \begin{equation}
        \alpha_{1}>\dots>\alpha_{P}
    \end{equation}
We now present the main theorem of the section, in which we prove the existence of a critical threshold $\alpha^{\star}$ for the stability of the retrievable memories of the IDP Hopfield model.
\begin{theorem}[Critical saliency $\alpha^{\star}$]
    Consider the dynamical system~\eqref{HopU} and let $\alpha_{1}>1$. Then
    \begin{enumerate}
        \item  there exists at least one locally exponentially stable equilibrium,
        \item let $\alpha^{\star}:=\frac{\gamma^\star}{\psi(\gamma^\star)}$ where $\gamma^*>0$ is such that $\psi'(\gamma^{\star})=\frac{1}{\alpha_{1}}$. Then
        \begin{equation}\label{eq:critical-threshold}
            x_{\rho}=\gamma_{\rho}\xi^{\rho}\ \text{is a locally exponentially stable equilibrium}\qquad\iff\qquad{\alpha_{\rho}}>\alpha^{\star}
        \end{equation}
    \end{enumerate}
\end{theorem}
\begin{proof} \ 
    \begin{enumerate} 
        \item  We define $\alpha(\gamma) := \gamma/\psi(\gamma)$ and let $\gamma_{1}>0$ be the positive solution of 
            \begin{equation}\label{eq: alp-1}
                \alpha_{1}=\alpha(\gamma_{1})=\frac{\gamma_{1}}{\psi(\gamma_{1})}
            \end{equation}
        and observe that
            \begin{equation}\label{eq: alp-dif}
                \alpha(\gamma)<\frac{1}{\psi'(\gamma)}\quad\forall\gamma>0
            \end{equation}
        As a matter of fact, it holds 
            \begin{equation}
                \frac{1}{\psi'(\gamma)}-\alpha(\gamma)=\frac{\psi(\gamma)-\gamma\psi'(\gamma)}{\psi'(\gamma)\psi(\gamma)}>0\quad{\forall{\gamma}>{0}}
            \end{equation}
        since $\psi'(\gamma)$, $\psi(\gamma)$, and $f(\gamma):=\psi(\gamma)-\gamma\psi'(\gamma)$ are strictly positive for all $\gamma> 0$. The fact that $f(\gamma)>0$ for all $\gamma>0$ follows by noting that $f(0)=0$ and $f'(\gamma)=-\gamma\psi''(\gamma)>0$, $\forall\gamma\ne 0$. 
         Then it follows from~\eqref{eq: alp-1}, \eqref{eq: alp-dif}] that        
            \begin{equation}
                \psi'(\gamma_{1})<\frac{1}{\alpha_{1}}
            \end{equation}        
        and thus, using Theorem \ref{thm: stb-cnd}, there exists at least one locally exponentially stable equilibrium point for~\eqref{HopU}.
        \item Since $\alpha_{1}>1$, then using~\eqref{eq: der},~\eqref{eq: asy},~\eqref{eq: 2der}] we have that $\exists{!}\gamma^{\star}>0$ such that
            \begin{equation}
                \psi'(\gamma^{\star})=\frac{1}{\alpha_{1}}
            \end{equation}
        Moreover, from~\eqref{eq: 2der} it is immediate to observe that        
            \begin{equation}
                \psi'(\gamma_{\rho})<\frac{1}{\alpha_{1}}\iff{\gamma_{\rho}>\gamma^{\star}}
            \end{equation}
        Notice now that 
        \begin{equation}
            \alpha'(\gamma)=\frac{f(\gamma)}{\psi(\gamma)^{2}}>0\quad\forall\gamma>0
        \end{equation}
        where $\alpha(\gamma)$ and $f(\gamma)$ have been defined in the proof of the first statement. The previous equation implies that $\alpha(\gamma)$ is monotonically increasing for $\gamma>0$. Consequently, the stability condition of Theorem \ref{thm: stb-cnd} is satisfied if and only if $\alpha_{\rho}=\alpha(\gamma_{\rho})>\alpha(\gamma^{\star})=\alpha^{\star}$.
    \end{enumerate}
\end{proof}

We now want to exploit the definition of $\alpha^{\star}={\gamma^{\star}}/{\psi(\gamma^{\star})}$ to provide the reader with an heuristic of how a strong dominance of a saliency weight $\alpha_{1}$ may disrupt the stable retrieval of any memory but $x_{1}=\gamma_{1}\xi^{1}$. Since we know that $\gamma^{\star}>0$ is the solution of the equation
    \begin{equation}
        \psi'(\gamma^{\star})=\frac{1}{\alpha_{1}}
    \end{equation}
then as $\alpha_{1}=\alpha(\gamma_{1})$ increases, we have that $\psi'(\gamma^{\star})$ decreases, which by~\eqref{eq: 2der}] is possible only if $\gamma^{\star}$ is increasing as well. It is then evident from~\eqref{eq: der} and \eqref{eq: asy} that when $\gamma^{\star}$ increases, so does $\alpha^{\star}=\alpha(\gamma^{\star})$. Thus, we have proved that the growth of the stability threshold $\alpha^{\star}$ depends exclusively on the dominant saliency weight $\alpha_{1}$, and moreover that the two are proportionally related to each other. Consequently, the sole growth of the dominant saliency weight $\alpha_{1}$ is sufficient to push the stability threshold $\alpha^{\star}$ above all the other saliency weights, so that the only stable memory left is indeed $x_{1}=\gamma_{1}\xi^{1}$.

\begin{figure}[th!]
    \centering
    \includegraphics[width=.7\linewidth]{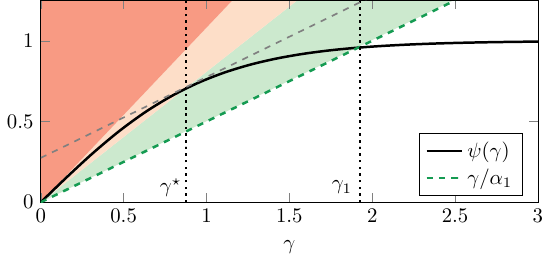}
    \caption{Visual representation of the existence and stability results presented in the previous sections. The saliency weights $\alpha$ such that the dissipation line $\gamma/\alpha$ belongs to: (1) the green region are associated to stable retrievable memories; (2) the yellow region are associated to unstable retrievable memories; (3) the red region to non-retrievable memories.}
    \label{fig:ActFun2}
\end{figure}

\section*{Energy of the IDP Hopfield Model and Energy Wells}
One of the most interesting features of the Hopfield model is the existence of an energy function that details how the trajectory of the system converge to the existing equilibria.
The energy of the IDP model is similar to the energy of the classic Hopfield model, under the assumption that $u\in{\real^{N}}$ is constant. Specifically, we introduce the following definition.

\begin{definition}[Energy of the IDP Hopfield model]
    Consider the dynamics~\eqref{HopU}]. Then the associated energy function is
    \begin{equation}\label{eq: energy_SI}
        \mathrm{E}(x;u)=-\frac{1}{2}\Psi(x)^{\top}W(u)\Psi(x)+x^{\top}\Psi(x)-\sum_{i=1}^{N}\int_{0}^{x_{i}}\psi_{i}(z)\:dz
    \end{equation}
    and we additionally define the energy per node as
    \begin{equation}
        \varepsilon(x;u)=\frac{\mathrm{E}(x;u)}{N}
    \end{equation}
\end{definition}
The following result characterizes the energy landscape of the IDP Hopfield model.
\begin{theorem}[Energy wells] Given a fixed input $u\in\real^{N}$, the energy function of the IDP Hopfield model~\eqref{HopU} has negative total time derivative
    \begin{equation}
        \frac{d\mathrm{E}}{dt}(x(t);u(t))|_{u(t)\equiv{u}}\leq{0}
    \end{equation}
    Moreover, given two equilibria $x_1=\gamma_1 \xi^1$, $x_2=\gamma_2 \xi^2$, $\gamma_1$, $\gamma_2>0$, associated to saliency weights $\alpha_{1},\alpha_{2}>1$ such that ${\alpha_{1}>\alpha_{2}}$, then
    \begin{equation}
        \varepsilon(x_{1};u)
        <\varepsilon(x_{2};u)
    \end{equation}
\end{theorem}
\begin{proof}
    Consider the total time derivative of the energy function associated to the IDP Hopfield model
    \begin{align}
        \frac{d\mathrm{E}}{dt}(x(t);u(t))|_{u(t)\equiv{u}}&=\frac{\partial{\mathrm{E}}}{\partial{x}}(x(t);u(t))|_{u(t)\equiv{u}}\frac{dx}{dt}(t)+\frac{\partial{\mathrm{E}}}{\partial{u}}(x(t);u(t))|_{u(t)\equiv{u}}\underbrace{\frac{du}{dt}(t)}_{\equiv{0}}\nonumber\\
        &=\frac{\partial{\mathrm{E}}}{\partial{x}}(x(t);u)\frac{dx}{dt}(t)=-\frac{dx}{dt}(t)^{\top}D\Psi(x(t))\frac{dx}{dt}(t)\leq{0}
    \end{align}
    where we have used the fact that $D\Psi(x)\succ{0}$, and consequently $\dot{\mathrm{E}}=0\iff{\dot{x}=0}$. \\We now want to explicitly compute the energy associated to a generic fixed point $x_{\rho}=\gamma_{\rho}\xi^{\rho},\ \gamma_{\rho}>0$, taking into account our homogeneity assumptions on the activation function. For notational simplicity, we will use $\mathrm{E}(x;u)=\mathrm{E}(x)$ and $\varepsilon(x,u)=\varepsilon(x)$    
    \begin{align}
        \mathrm{E}(x_{\rho})&=-\frac{1}{2}\Psi(x_{\rho})^{\top}W(u)\Psi(x_{\rho})+x_{\rho}^{\top}\Psi(x_{\rho})-\sum_{i=1}^{N}\int_{0}^{{x_{\rho}}_{i}}\psi(z)\:dz\nonumber\\
        &=-\frac{1}{2}\psi(\gamma_{\rho}){\xi^{\rho}}^{\top}W(u)\xi^{\rho}\psi(\gamma_{\rho})+\gamma_{\rho}{\xi^{\rho}}^{\top}\xi^{\rho}\psi(\gamma_{\rho})-N\int_{0}^{\gamma_{\rho}}\psi(z)\:dz\nonumber\\
        &=-\frac{N}{2}\alpha_{\rho}\psi(\gamma_{\rho})^{2}+N\gamma_{\rho}\psi(\gamma_{\rho})-N\int_{0}^{\gamma_{\rho}}\psi(z)\:dz\\
    \end{align}
    
    Considering now the fixed point condition
    \begin{equation}
        \frac{\gamma_{\rho}}{\alpha_{\rho}}=\psi(\gamma_{\rho})
    \end{equation}
    and switching to the energy per node    
    \begin{align}
        \varepsilon(x_{\rho})&=\frac{1}{2}\gamma_{\rho}\psi(\gamma_{\rho})-\int_{0}^{\gamma_{\rho}}\psi(z)\:dz\\
        &=\frac{1}{2}\frac{\gamma_{\rho}^{2}}{\alpha_{\rho}}-\int_{0}^{\gamma_{\rho}}\psi(z)\:dz
    \end{align}
    
    Now, considering
    \begin{equation}
        F(\gamma,\alpha)=\psi(\gamma)-\frac{\gamma}{\alpha}
    \end{equation}
    we know that $F(\gamma,\alpha_{\rho})>0$, ${\forall{\gamma}\in{(0,\gamma_{\rho})}}$, from which we obtain   
    \begin{align}
        \varepsilon(x_{\rho})&=\frac{1}{2}\frac{\gamma_{\rho}^{2}}{\alpha_{\rho}}-\int_{0}^{\gamma_{\rho}}\psi(z)\:dz\\
        &<\frac{1}{2}\frac{\gamma_{\rho}^{2}}{\alpha_{\rho}}-\int_{0}^{\gamma_{\rho}}\frac{z}{\alpha_{\rho}}\:dz=0
    \end{align}
    
    Considering now $\alpha_{1},\alpha_{2}>1,\ {\alpha_{1}}>\alpha_{2}$, we know that $F(\gamma_{2},\alpha_{2})=F(\gamma_{1},\alpha_{1})=0$. Since $F$ is monotonically increasing in $\alpha$ and decreasing in $\gamma$, then $\gamma_{1}>\gamma_{2}$. Therefore, it follows that
    \begin{align}
        \varepsilon(x_{1})-\varepsilon(x_{2})&=\frac{1}{2}\bigg{(}\frac{\gamma_{1}^{2}}{\alpha_{1}}-\frac{\gamma_{2}^{2}}{\alpha_{2}}\bigg{)}-\int_{\gamma_{2}}^{\gamma_{1}}\psi(z)\:dz\\
        &<\frac{1}{2}\bigg{(}\frac{\gamma_{1}^{2}}{\alpha_{1}}-\frac{\gamma_{2}^{2}}{\alpha_{2}}\bigg{)}-\int_{\gamma_{2}}^{\gamma_{1}}\frac{z}{\alpha_{1}}\:dz\\
        &=-\frac{1}{2}\bigg{(}\frac{\gamma_{2}^{2}}{\alpha_{2}}-\frac{\gamma_{2}^{2}}{\alpha_{1}}\bigg{)}<0
    \end{align}
\end{proof}
From Fig.~\ref{fig:Energy} it is possible to observe how the increase of a specific saliency weight significantly deepens the energy well associated to the corresponding memory. The varying depth of the energy wells provides a clear graphical interpretation of the stability of the different memories encoded in the IDP Hopfield model. Specifically, considering panel Fig.~\ref{fig:Energy}(d), it is easy to observe how an appropriate choice of a perturbation $\delta{u}\in\real^{N}$ for the dynamics~\eqref{HopU} may easily push the state trajectory from the flatter to the deeper energy well. At the same time, the same perturbation $\delta{u}$ may be too weak to reverse the previous jump, thus leaving the state trajectory inevitably confined to the basin of attraction of the memory associated to the deeper energy well.

\begin{figure}[th!]
    \centering
    \includegraphics[width=.75\linewidth]{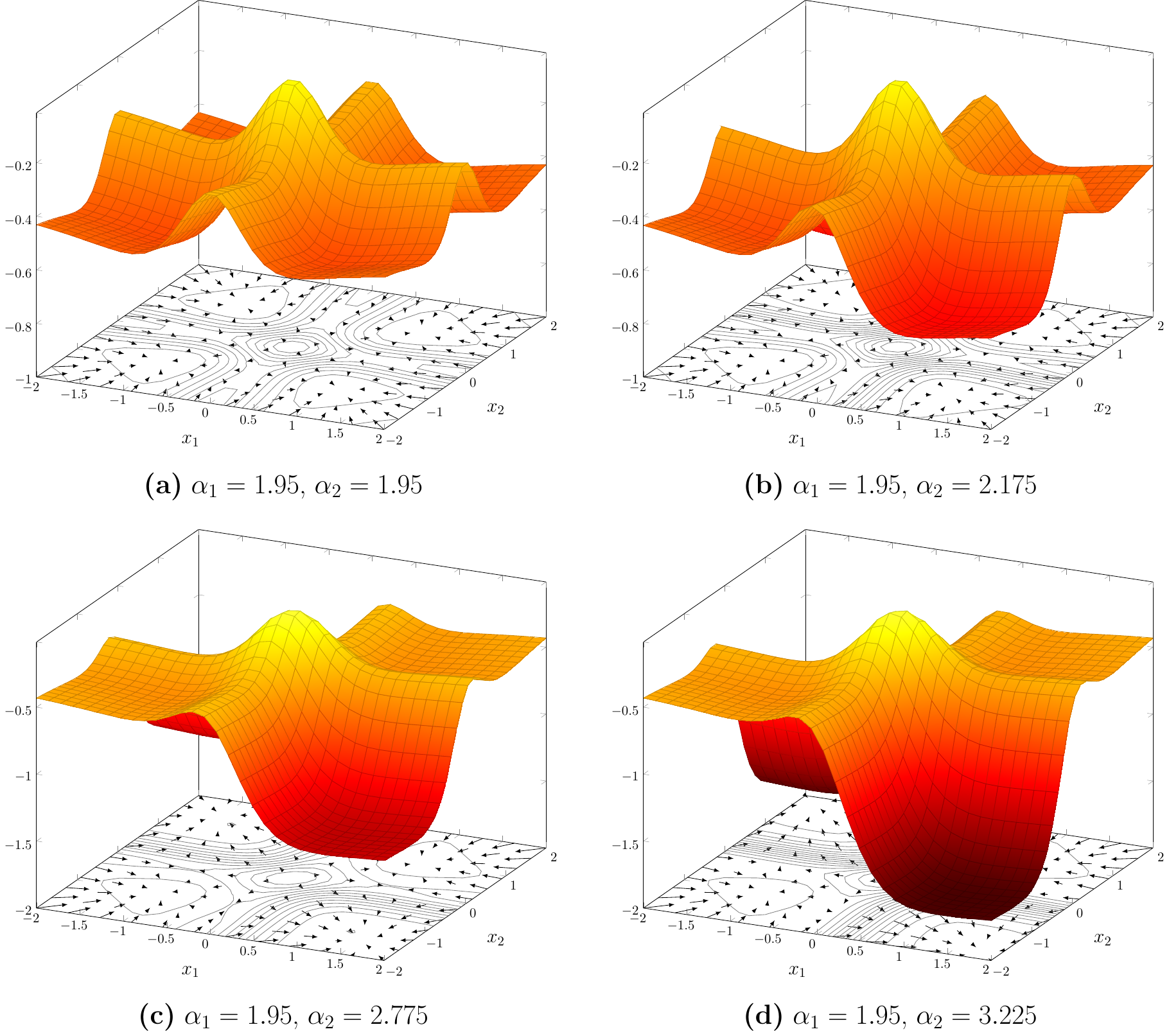}
    \caption{Input-driven modulation of the energy landscape of a $2$-d Hopfield model having orthogonal memories $\xi^{1}=[1\ 1]$ and $\xi^{2}=[1\ -1]$. (a) Energy landscape with $\alpha_{1}=\alpha_{2}$, which reduces to the case of the classic Hopfield model. As can be observed, the depth of the minima associated to the memories $\xi^{1}$ and $\xi^{2}$ is the same. (b-d) Energy landscape with constant $\alpha_{1}$ and increasing $\alpha_{2}$. As can be observed, the depth of the minima associated to $\xi^{1}$ remains unvaried, while the depth of the minima associated to memory $\xi^{2}$ significantly deepens as $\alpha_{2}$ grows.} 
    \label{fig:Energy}
\end{figure}

\section*{Stochastic Dynamics for the Classic and IDP Hopfield}
We now want to exploit some classical results on stochastic dynamics to understand how the models perform retrieval when embedded in an environment disrupted by noise. Consider a system of the type
\begin{equation}\label{StocSys}
    \begin{cases}
        \dot{x}(t)=\mathcal{F}(x)+G(x)\eta(t)\\
        x(0)=x_{0}\in{\real^{N}}
    \end{cases}
\end{equation}
where $G\in{C^{2}(\real^{N};\real^{N\times{N}}})$ and the noise obeys the following statistics
\begin{align}
    \mathbb{E}[\eta(t)]&\equiv{0}\in{\real}^{N}\\
    \mathbb{E}[\eta(t)\eta(\tau)]&=\sigma^{2}I_{N}\delta(t-\tau)
\end{align}

The probability density function $\mathbb{P}(t,x)$ associated with the evolution of the trajectories of the stochastic system~\eqref{StocSys} is described by the following parabolic P.D.E., known in the literature as Fokker-Planck Equation
\begin{equation}\label{FPE_SI}
        \frac{\partial{\mathbb{P}}}{\partial{t}}(t,x)=-\sum_{i=1}^{N}\frac{\partial}{\partial{x_{i}}}[\mathcal{F}_{i}(x)\mathbb{P}(t,x)]\!+\!\frac{\sigma^{2}}{2}\!\sum_{i,j=1}^{N}\frac{\partial^{2}}{\partial{x_{i}\partial{x_{j}}}}[D_{ij}(x)\mathbb{P}(t,x)]
\end{equation}
where $D(x)=G(x)G(x)^{\top}$ is the state-dependent covariance matrix for the noise increments. In our setting, we expressly want to avoid any assumption on the state-dependent structure of the noise, and for this reason we consider the following
    \begin{equation}
        \dot{x}(t) = \mathcal{F}(x) + \eta(t)
    \end{equation}
where we choose the Ito formalism for the S.D.E. simulation. Notice that for both the classic and IDP Hopfield model, the drift term can be written as $\mathcal{F}(x)=D\Psi(x)^{-1}\nabla_{x}\mathrm{E}(x,W)$, and consequently our S.D.E. becomes
    \begin{equation}\label{eq: SDEsim}
        \dot{x} = D\Psi(x)^{-1}\nabla_{x}\mathrm{E}(x,W) + \eta(t)
    \end{equation}
Given that $D\Psi(x)$ is a positive definite diagonal matrix, so is $D\Psi(x)^{-1}$, and therefore it can be used as a tailored state-dependent covariant matrix by choosing $G(x)=D\Psi(x)^{-1/2}$. Following this observation, Wong \citep{wong1991stochastic} proposed the model known as diffusion state machine
    \begin{equation}\label{eq: dif-mac}
        \dot{x}(t) = -D\Psi(x)^{-1}\nabla_{x}\mathrm{E}(x,W)+\sigma^{2}\nabla\cdot(D\Psi(x)^{-1}) + \sqrt{2}{D\Psi(x)^{-1/2}}\eta(t)
    \end{equation}
for which the associated Fokker-Planck equation admits a closed form stationary solution
    \begin{equation}\label{eq: SPD}
        \mathbb{P}_{\infty}(x)=Z^{-1}e^{-\frac{E(x;W)}{\sigma^{2}}}
    \end{equation}

Notice that the stationary probability density~\eqref{eq: SPD} strictly depends on the energy~\eqref{eq: energy_SI}, and consequently most of the probability density will concentrate around the deepest minima. A quick comparison of Fig.~\ref{fig:PDF}(a-b) and Fig.~\ref{fig:PDF}(c-d) respectively highlights how the solutions of~\eqref{FPE_SI} in the simple case of state-independent noise and of the diffusion state machine display very similar density profiles. Consequently, although unable to provide a closed form solution of the stationary Fokker-Planck equation associated to~\eqref{eq: SDEsim}, it is reasonable to suppose that even for simple state-independent noise the stationary probability density will closely depend on the energy~\eqref{eq: energy_SI} of the Hopfield model.
In the displayed $2$-d case, multiplication by a radially decaying kernel has been necessary to make the closed form stationary probability density compatible with the vanishing boundary conditions imposed to solve the numerical problem. Namely
    \begin{equation}
        \mathbb{P}_{\infty}(x_{1},x_{2})=\frac{1}{\mathcal{Z}}\left[\exp\left(-\frac{\mathrm{E}((x_{1},x_{2});W)}{\sigma^{2}}\right)\exp\left(-\frac{x_{1}^{2}+x_{2}^{2}}{c}\right)\right]\quad{c>0}
    \end{equation}
where $\mathcal{Z}>0$ is a new normalizing constant over the area of $\Omega\subset{\real}^{2}$, now with $(x_{1},x_{2})\in\Omega$ coordinates.

\begin{figure}[th!]
    \centering
    \includegraphics[width=.85\linewidth]{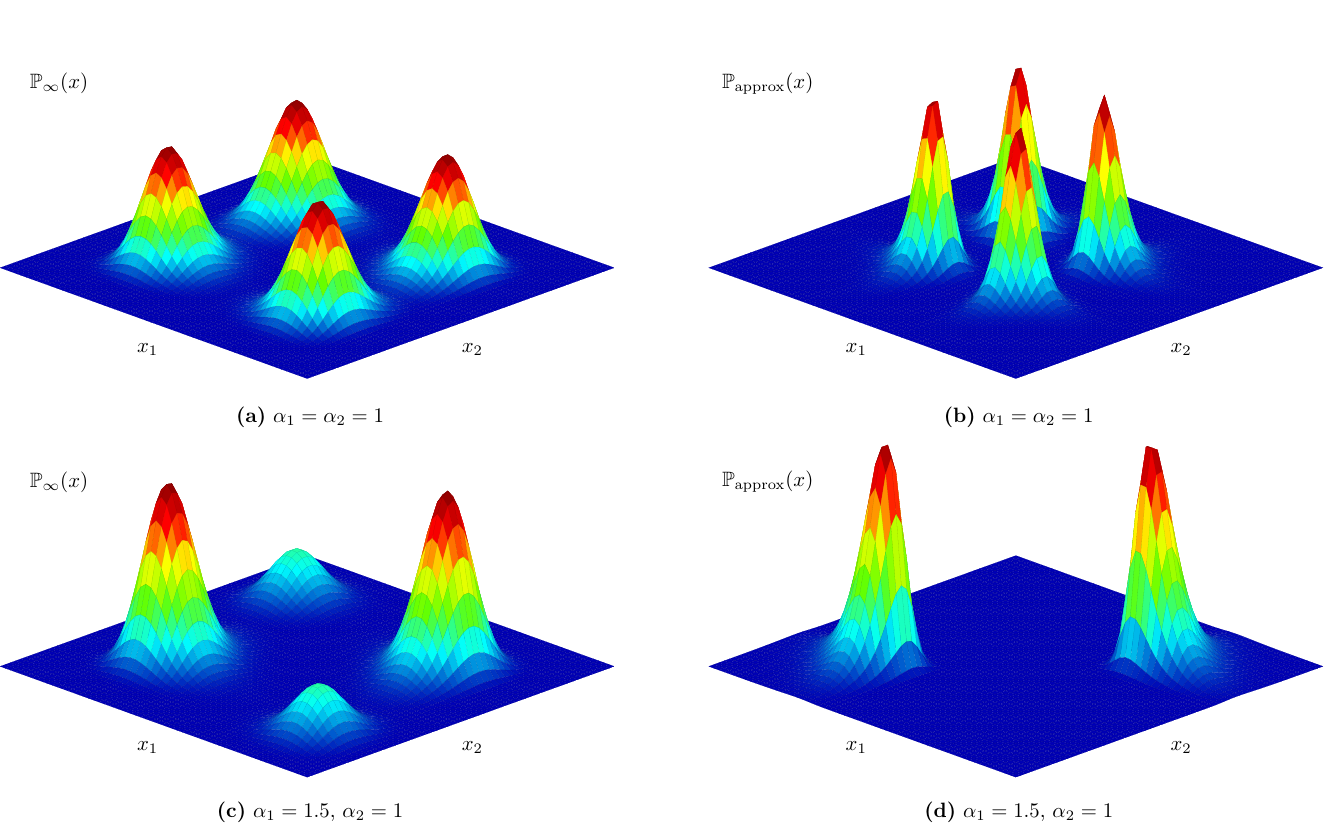}
    \caption{Numerical solution of the Fokker Planck equation and plotting of the stationary probability density~\eqref{eq: SPD} for the (Panels (a)-(b)) classical Hopfield model and (Panels (c)-(d)) IDP Hopfield model, where in the latter case the small probability mass around the shallow minima has been totally absorbed by the deepest minima.}
    \label{fig:PDF}
\end{figure}

We now consider a numerical comparison in the retrieval performance of the classic Hopfield model and the IDP Hopfield model. Both models where instantiated with a network size of $N=1024$ to emulate a large neuronal cluster, and an interval of $t=10$ simulation time was allotted to each network for the retrieval of the desired memory. The chosen amplitude for the noise was $\sigma=8$, and the activation function was a swiftly saturating hyperbolic tangent $\psi(\cdot)=\tanh(10\ \cdot)$. As observable from Fig.~\ref{fig:comparison-additive-multiplicative}(a), the activity in the classic Hopfield network aligns with the provided input during presentation, but then is quickly disrupted by the noise at input offset and none of the memories is retrieved. Instead, from Fig.~\ref{fig:comparison-additive-multiplicative}(b) it is clear that the activity of the IDP Hopfield network quickly aligns with the desired memory, and efficiently switches in accordance with the input characterization. Notice that for lower noise amplitudes $\sigma$ the classic Hopfield model is still able to perform correct retrieval, but performances consistently worsen as the noise amplitude increases.

\begin{figure}[th!]
    \centering
    \includegraphics[width=1\linewidth]{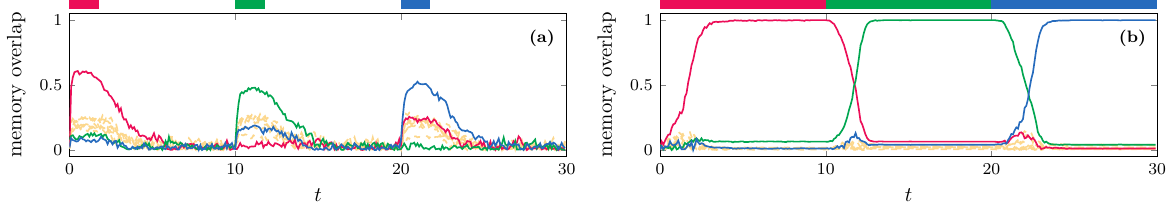}
    \caption{Comparison of the retrieval performance for the classic vs IDP Hopfield model with white noise. (a) In the classic Hopfield model, with the noise quickly disrupts the transient towards any of the memories stored in the synaptic matrix. (b) In the IDP Hopfield model, the noise is instrumental for a swift transition to the memory associated with the dominant saliency weight in the input.}
    \label{fig:comparison-additive-multiplicative}
\end{figure}

In addition, we can observe from Fig.~\ref{fig:comparison-additive-multiplicative}(b) that the transients from one memory to the next are quite smooth, and that the network activity seems to be shortly entrapped in the previous recall at input switching. Thus, it seemed natural to devise an experiment where the input shortly switches its dominant component, a situation we refer to as ``glitch'', and which represents a decoding error from the agent. Surprisingly, we see from Fig.~\ref{fig:glitchy}(b) that the IDP Hopfield model is effective at correcting the glitch on the base of the previous recall and the time length of the glitchy input. Thus, the new IDP model is robust both against environmental noise endogenous to the network and exogenous interference arising from errors in the decoding process. 

\begin{figure}[th!]
    \centering
    \includegraphics[width=1\linewidth]{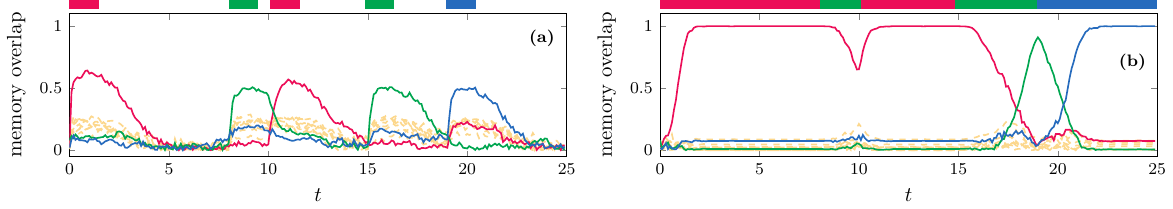}
    \caption{Comparison of the retrieval performance for the classic vs IDP Hopfield model with noise and a glitchy input. (a) Similarly the Fig.~\ref{fig:comparison-additive-multiplicative}(a) the classic Hopfield model activity briefly aligns with the input, but no memory is actually retrieved due to the interference from the noise. (b) The IDP Hopfield model displays the remarkable capability of correcting decoding errors from the input pathways, provided that the error does not last long enough. Specifically, it can be observed that for $t\in[8,10]$ the input changes its dominant component, but the network is able to maintain fixation on the previous recall. Instead, the fixation on the previous recall is broken when the same error is presented for $t\in[15,20]$.}
    \label{fig:glitchy}
\end{figure}

\section*{IDP Framework for the Firing Rate Model}
The firing rate model of associative memory is an alternative to the Hopfield model for the modeling of memories as point attractors in a recurrent neural network. Differently from the Hopfield model, which draws inspiration from the spin glass models of statistical mechanics, the firing rate model strives to offer a more biologically plausible account of memory processes in the brain. The key features that make this model biologically plausible are its positivity and the fact that prototypical memories can be selected as sparse $0$-$1$ vectors, capturing realistic low firing rate activity \citep[Sec. 17.2.6]{gerstner2014neuronal}. 

We now provide more details on the firing rate model of associative memory and expand by means of the IDP paradigm. Define the expected average activity parameter $p\in(0,1)$ and defined random i.i.d. memory vectors $\xi^{1},\dots,\xi^{P}$ with entries
    \begin{equation}
        \mathbb{P}[\xi^{\mu}_{i}=1]=p\qquad\mathbb{P}[\xi^{\mu}_{i}=0]=1-p\qquad\forall\mu=1,\dots,P,\ \forall{i}=1,\dots,N
    \end{equation}
Notice that for neurologically accurate firing rates, the expected average activity parameter is fixed at $p\approx{0.2}$ \citep{gastaldi2021shared}, hence to a low activity regime. By comparison, if we transpose the Hopfield model to have memories in $\{0,1\}$, then the expected average activity parameter would be $p=0.5$. The synaptic matrix is then constructed as in \citep[Sec. 7.4]{dayan2005theoretical} with values
    \begin{equation}
        W = \frac{c}{Np(1-p)}\sum_{\mu=1}^{P}\left(\xi^{\mu}-p\vectorones\right)\left(\xi^{\mu}-p\vectorones\right)^{\top}-\frac{1}{Np}\vectorones\vectorones^{\top}
    \end{equation}
The dynamics of the firing rate model are then defined as follow
    \begin{equation}\label{eq: FR-A}
        \begin{cases}
            \dot{s}(t)=-s(t)+\Gamma(Ws(t) + u(t))\\
            s(0)=s_{0}\in[0,1]^{N}
        \end{cases}
    \end{equation}
where the activation function is diagonal and homogeneous and where $u(t)\in\real^{N}$ is a time-varying external input. In our case, we have chosen
    \begin{equation}
        \Gamma_{i}(z)=\frac{1}{1+e^{-\rho(z_{i}-1/2)}}\quad\forall{i=1,\dots,N},\ z\in\real^{N},
    \end{equation}
the classical sigmoid function centered at $z_{i}=1/2$. Notice that
$\Gamma_{i}:\real\rightarrow(0,1)$, $\forall{i=1,\dots,N}$, and for our
particular choice of initial conditions we have by Nagumo's theorem
\citep{Blanchini2015} that the firing rate dynamics are confined in the
hyper-interval $[0,1]^{N}$.

The dynamical system~\eqref{eq: FR-A} is the classical firing rate model with an additive external input. We are now interested in understanding the corresponding IDP formulation of the firing rate model in light of what has been done for the Hopfield model.
\begin{definition}
    The IDP formulation of the firing rate model for a network of $N\in\mathbb{N}$ units and external input $u(t)\in\real^{N}$ is
        \begin{equation}\label{eq: FR-M}
            \begin{cases}
                \dot{s}(t)=-s(t)+\Gamma(W(u(t))s(t))\\
                s(0)=s_{0}\in[0,1]^{N}
            \end{cases}
        \end{equation}
    where the input-modulated synaptic matrix is
        \begin{align}
            W(u) &= \frac{c}{Np(1-p)}\sum_{\mu=1}^{P}\alpha_{\mu}(u)\left(\xi^{\mu}-p\vectorones\right)\left(\xi^{\mu}-p\vectorones\right)^{\top}-\frac{1}{Np}\vectorones\vectorones^{\top}\\
            \alpha_{\mu}(u)&=\frac{1}{Np}{\xi^{\mu}}^{\top}u(t)
        \end{align}
\end{definition}
The presented model has been studied only numerically for fairly mixed inputs, i.e., inputs $u=\sum_{\mu=1}^{P}\alpha_{\mu}\xi^{\mu}$ where the saliency weights $\alpha$'s are of the same order of magnitude. We now present a figure depicting the overlaps between each of the models memory and the activity of the network, formally defined as
    \begin{equation}
        m_{\mu}(t)=\frac{1}{Np}s(t)^{\top}\xi^{\mu}\quad{\mu=1,\dots,P}
    \end{equation}

\begin{figure}[th!]
    \centering
    \includegraphics[width=1\linewidth]{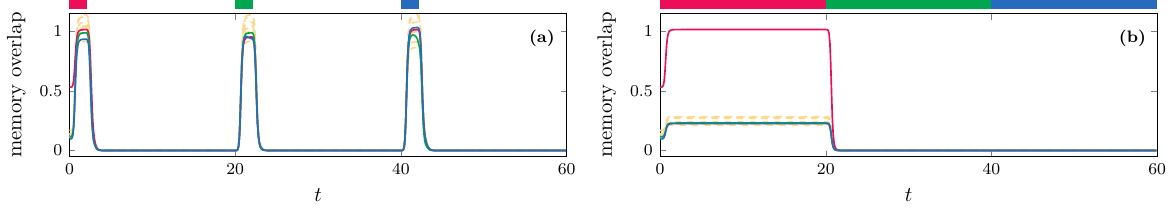}
    \caption{Overlaps for the network memories and the network activity in the case of very mixed inputs. The saliency weights $\alpha$'s are drawn from a uniform distribution $\mathcal{U}[0,4]$, and the dominant weight is ad hoc multiplied by a constant so that is at least 3 times larger than all other saliency weights. (a) Overlap activity for the firing rate system~\eqref{eq: FR-A}. As observable, the activity is shot in a state that is in perfect overlap with the memories (probably along the bisector) and then quickly decays to the inactivation state $\vectorzeros$. (b) Overlap activity for the firing rate system~\eqref{eq: FR-M}. As observable, the network reliably recovers the first memory, but when the synaptic matrix is altered to favor a transition towards the second memory, the network activity remains entrapped in the inactivation state $\vectorzeros$.}
    \label{fig:NoNoiseTraj}
\end{figure}

From Fig.~\ref{fig:NoNoiseTraj} it is clear that the inactivation state $\vectorzeros$ is a highly attractive spurious state of the system. Probably, due to the very highly mixed nature of the input and its dominance over synaptic dynamics, at input onset the activity of the additive firing rate system is shot right along the bisector of the hyper-interval $[0,1]^{N}$. This point in state space is probably well separated from the basin of attraction of any individual memory, as one can see from the quick decay to inactivation at input offset. Instead, we see that the IDP firing rate model is capable of correctly retrieving the first memory, but then inevitably falls into inactivation during the transient to the second memory. The falls into inactivation could well be a further proof of the width of the basin of attraction associated to the spurious $\vectorzeros$. Therefore, it seems natural to wonder whether small endogenous perturbations are enough to drive the system from the inactivation state to the desired memory. The stochastic dynamics associated to the IDP firing rate model are
    \begin{equation}\label{eq: FR-M-S}
            \begin{cases}
                \dot{s}(t)=-s(t)+\Gamma(W(u(t))s(t))+\eta(t)\\
                s(0)=s_{0}\in[0,1]^{N}
            \end{cases}
        \end{equation}
where the white noise $\eta(t)$ follows the statistics
    \begin{align}
        \mathbb{E}[\eta(t)]&\equiv{0}\\
        \mathbb{E}[\eta(t)\eta(\tau)^{\top}]&=\sigma^{2}I_{N}\delta(t-\tau)
    \end{align}
    
\begin{figure}[th!]
    \centering
    \includegraphics[width=1\linewidth]{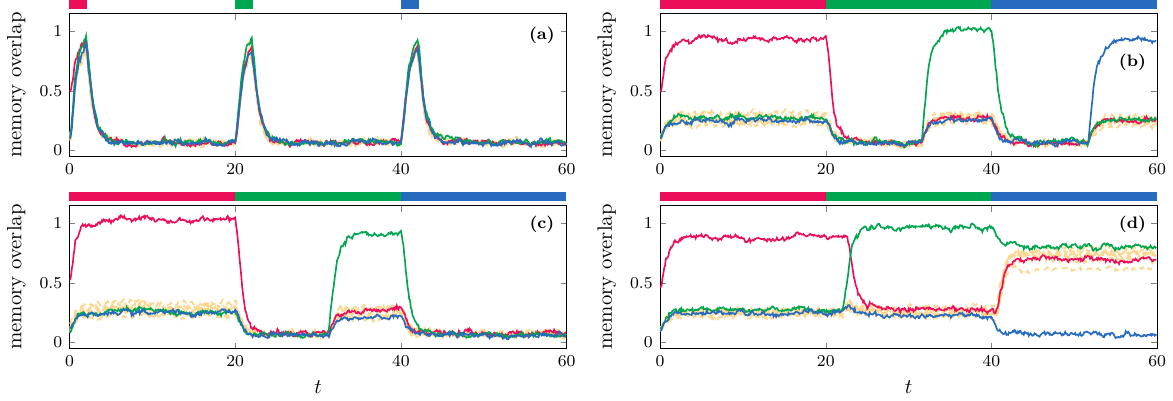}
    \caption{Overlaps for the network memories and the network activity in the case of very mixed inputs and the influence of white noise. (a) Overlap activity for the firing rate system~\eqref{eq: FR-A} with the addition of white noise. As observable, the overlap activity is the same as that of Fig.~\ref{fig:NoNoiseTraj}.(a), with the network activity still being shot along the bisector of $[0,1]^{N}$ at input onset. At input offset, the noise amplitude is not strong enough to drive the system towards one of the memories. (b-d) Overlap activity for the IDP firing rate system~\eqref{eq: FR-M-S}. In panel (b), the noise in combination with the synaptic modulation is functional at inducing the correct switching among memories after a brief period of inactivity. Instead, panel (c-d) present less ideal situation where switching between memories fails in some way.}
    \label{fig:NoiseTraj}
\end{figure}

The IDP firing rate model has been investigated only numerically, and as illustrated in Fig.~\ref{fig:NoiseTraj}.(b-d) there is high variability in the retrieval performance. This hints at the fact that, due to the activity being confined in a relatively small hyper-interval, the model is very sensitive to parameters changes and external influences. Future work should address the theoretical aspect of the model.

\end{document}